\documentclass[aos,preprint]{imsart}
\usepackage{xr}
\usepackage{latexsym}
\usepackage{amsmath}
\usepackage{amssymb}
\usepackage{amsthm}
\usepackage{amsfonts}
\usepackage{graphicx}
\usepackage{fancyhdr}
\usepackage{indentfirst}
\usepackage{makeidx}
\usepackage{amsfonts}
\usepackage{verbatim}
\usepackage{hyperref}
\usepackage{graphicx,amsmath,amssymb,epsfig,float}
\usepackage{amsfonts}
\usepackage{booktabs}
\usepackage{subcaption}
\usepackage{verbatim}

\RequirePackage[OT1]{fontenc}
\RequirePackage[colorlinks]{hyperref}
\renewcommand{\qed}{\alignleft {\text{\tiny{$\blacksquare$}}}}
\theoremstyle{plain}
\newtheorem{theorem}{Theorem}[section]
\newtheorem{lemma}{Lemma}[section]
\newtheorem{corollary}{Corollary}[section]

\theoremstyle{definition}

\newtheorem{remark}{Remark}[section]
\newtheorem{example}{Example}[section]
\arxiv{arXiv:1406.4643}
\startlocaldefs

\newcommand\R{{\mathbb R}}
\newcommand\E{{\mathrm E}}
\renewcommand\P{{\mathrm P}}

\newcommand\PPP{{\cal{P}}}
\newcommand\NNN{{\mathbb N}}
\renewcommand\tilde{\widetilde}

\def\<#1,#2>{\left<#1,#2\right>}
\let\bar\overline

\numberwithin{equation}{section}
\endlocaldefs
\input{tcilatex}
\begin{document}

\begin{frontmatter}

\title{Vector Quantile Regression: An Optimal Transport Approach\thanksref{T1}}
\runtitle{Vector Quantile Regression}
\thankstext{T1}{First ArXiv submission: June, 2014. Carlier is partially supported by
INRIA and the ANR for partial support through the Projects ISOTACE
(ANR-12-MONU-0013) and OPTIFORM (ANR-12-BS01-0007). Chernozhukov
is partially supported by an NSF grant. Galichon is partially supported by the European Research Council under the European
Union's Seventh Framework Programme (FP7/2007-2013) / ERC grant agreement
313699.}

\begin{aug}
\author{\fnms{Guillaume} \snm{Carlier}\thanksref{t1}\ead[label=e1]{carlier@ceremade.dauphine.fr}},
\author{\fnms{Victor} \snm{Chernozhukov}\thanksref{t2}\ead[label=e2]{vchern@MIT.edu}}
\and
\author{\fnms{Alfred} \snm{Galichon}\thanksref{t3}
\ead[label=e3]{ag133@nyu.edu}}

\runauthor{G. Carlier et al.}

\affiliation{\thanksmark{t1}Universit\'e Paris IX Dauphine, \thanksmark{t2}MIT, and \thanksmark{t3}NYU and Sciences Po}

\address{CEREMADE, UMR CNRS 7534, \\
Universit\'e Paris IX Dauphine\\
Pl. de Lattre de Tassigny, \\
75775 Paris Cedex 16, FRANCE\\
\printead{e1}}

\address{Department of Economics \\
and Center for Statistics, MIT\\
50 Memorial Drive, E52-361B\\
Cambridge, MA 02142, USA\\
\printead{e2}}

\address{Economics Department \\
and Courant Institute of Mathematical Sciences, NYU\\
70 Washington Square South\\
New York, NY 10013, USA\\
\printead{e3}}
\end{aug}

\begin{abstract}
We propose a notion of conditional vector quantile function and a vector
quantile regression. A \emph{conditional vector quantile function} (CVQF) of a random vector $Y$, taking
values in $\mathbb{R}^d$ given covariates $Z=z$, taking values in $\mathbb{R}%
^k$, is a map $u \longmapsto Q_{Y\mid Z}(u,z)$, which is monotone, in the sense
of being a gradient of a convex function, and such that given that vector $U$
follows a reference non-atomic distribution $F_U$, for instance uniform
distribution on a unit cube in $\mathbb{R}^d$, the random vector $Q_{Y\mid
Z}(U,z)$ has the distribution of $Y$ conditional on $Z=z$. Moreover, we have
a strong representation, $Y = Q_{Y\mid Z}(U,Z)$ almost surely, for some
version of $U$. The \emph{vector quantile regression} (VQR) is a linear model for CVQF of $Y$ given
$Z$. Under correct specification, the notion produces strong representation,
$Y=\beta \left(U\right) ^\top f(Z)$, for $f(Z)$ denoting a known set of
transformations of $Z$, where $u \longmapsto \beta(u)^\top f(Z)$ is a monotone
map, the gradient of a convex function, and the quantile regression
coefficients $u \longmapsto \beta(u)$ have the interpretations analogous to that
of the standard scalar quantile regression. As $f(Z)$ becomes a richer class
of transformations of $Z$, the model becomes nonparametric, as in series
modelling. A key property of VQR is the embedding of the classical
Monge-Kantorovich's optimal transportation problem at its core as a special
case. In the classical case, where $Y$ is scalar, VQR reduces to a version of the
classical QR, and CVQF reduces to the scalar conditional quantile function.
An application to  multiple Engel curve
estimation is considered.
\end{abstract}



\end{frontmatter}

\section{Introduction}

Quantile regression provides a very convenient and powerful tool for
studying dependence between random variables. The main object of modelling
is the conditional quantile function (CQF) $(u, z) \longmapsto Q_{Y\mid Z}(u,z)$%
, which describes the $u$-quantile of the random scalar $Y$ conditional on a
$k$-dimensional vector of regressors $Z$ taking a value $z$. Conditional
quantile function naturally leads to a strong representation via relation:
\begin{equation*}
Y = Q_{Y\mid Z}(U, Z), \ \ U\mid Z \sim U(0,1),
\end{equation*}
where $U$ is the latent unobservable variable, normalized to have a uniform
reference distribution, and is independent of regressors $Z$. The mapping $u
\longmapsto Q_{Y\mid Z} (u, Z)$ is monotone, namely non-decreasing, almost
surely.%

Quantile regression (QR) is a means of modelling the conditional quantile
function. A leading approach is linear in parameters, namely, it assumes
that there exists a known ${\mathbb{R}}^{p}$-valued vector $f(Z)$,
containing transformations of $Z$, and a ($p\times 1$ vector-valued) map of
regression coefficients $u\longmapsto \beta (u)$ such that
\begin{equation*}
Q_{Y\mid Z}\left( u\mid z\right) =\beta \left( u\right) ^{\top }f(z),
\end{equation*}%
for all $z$ in the support of $Z$ and for all quantile indices $u$ in $(0,1)$%
. This representation highlights the vital ability of QR to capture
differentiated effects of the explanatory variable $Z$ on various
conditional quantiles of the dependent variable $Y$ (e.g., impact of
prenatal smoking on infant birthweights). QR has found a large number of
applications; see references in Koenker (\cite{Koenker2005})'s monograph.
The model is flexible in the sense that, even if the model is not correctly
specified, by using more and more suitable terms $f(Z)$ we can approximate
the true CQF arbitrarily well. Moreover, coefficients $u\longmapsto \beta (u)$
can be estimated via tractable linear programming method (\cite{KB}).

The principal contribution of this paper is to extend these ideas to the
cases of vector-valued $Y$, taking values in $\mathbb{R}^{d}$. Specifically,
a vector conditional quantile function (CVQF) of a random vector $Y$, taking
values in $\mathbb{R}^{d}$ given the covariates $Z$, taking values in $%
\mathbb{R}^{k}$, is a map $(u,z)\longmapsto Q_{Y\mid Z}(u,z)$, which is monotone
with respect to $u$, in the sense of being a gradient of a convex function,
which implies that%
\begin{equation}
(Q_{Y\mid Z}(u,z)-Q_{Y\mid Z}(\bar{u},z))^{\top }(u-\bar{u})\geq 0\ \ \text{
for all }u,\bar{u}\in (0,1)^{d},z\in \mathcal{Z},  \label{eq: monotonicity}
\end{equation}%
and such that the following strong representation holds with probability 1:
\begin{equation}
Y=Q_{Y\mid Z}(U,Z),\ \ U\mid Z\sim U(0,1)^{d},  \label{strong0}
\end{equation}%
where $U$ is latent random vector uniformly distributed on $(0,1)^{d}$. We
can also use other non-atomic reference distributions $F_{U}$ on $\mathbb{R}%
^{d}$, for example, the standard normal distribution instead of uniform
distribution (as we can in the canonical, scalar quantile regression case).
We show that this map exists and is unique under mild conditions, as a
consequence of Brenier's polar factorization theorem. This notion relies on
a particular, yet very important, notion of monotonicity (\ref{eq:
monotonicity}) for maps $\mathbb{R}^{d}\rightarrow \mathbb{R}^{d}$, which we
adopt here.%

We define vector quantile regression (VQR) as a model of CVQF, particularly
a linear model. Specifically, under correct specification, our linear model
takes the form:
\begin{equation*}
Q_{Y\mid X}(u\mid z)=\beta (u)^{\top }f(z),
\end{equation*}%
where $u\longmapsto \beta (u)^{\top }f(z)$ is a monotone map, in the sense of
being a gradient of convex function; and $u\longmapsto \beta (u)$ is a map of
regression coefficients from $(0,1)^{d}$ to the set of $p\times d$ matrices
with real entries. This model is a natural analog of the classical QR for
the scalar case. In particular, under correct specification, we have the
strong representation
\begin{equation}
Y=\beta (U)^{\top }f(Z),\ \ U\mid Z\sim U(0,1)^{d},
\label{eq: linear factor}
\end{equation}%
where $U$ is uniformly distributed on $(0,1)^{d}$ conditional on $Z$. (Other
reference distributions could also be easily permitted.)

We provide a linear program for computing $u\longmapsto \beta (u)$ in population
and finite samples. We shall stress that this formulation offers a number of
useful properties. In particular, the linear programming problem admits a
general formulation that embeds the optimal transportation problem of
Monge-Kantorovich-Brenier, establishing a useful conceptual link to an
important area of optimization and functional analysis (see, e.g. \cite%
{Villani03}, \cite{villani2}).

Our paper also connects to a number of interesting proposals for performing
multivariate quantile regressions, which focus on inheriting certain (though
not all) features of univariate quantile regression-- for example,
minimizing an asymmetric loss, ordering ideas, monotonicity, equivariance or
other related properties, see, for example, some key proposals (including
some for the non-regression case) in \cite{chaudhuri}, \cite{Koltchinskii},
\cite{serfling}, \cite{HPS}, \cite{KongMizera}, \cite{Belloni:Wrinkler},
which are contrasted to our proposal in more details below. Note that it is
not possible to reproduce all "desirable properties" of scalar quantile
regression in higher dimensions, so various proposals focus on achieving
different sets of properties. Our proposal is quite different from all of
the excellent aforementioned proposals in that it targets to simultaneously
reproduce two fundamentally different properties of quantile regression in
higher dimensions -- namely the deterministic coupling property (\ref{eq:
linear factor}) and the monotonicity property (\ref{eq: monotonicity}). This
is the reason we deliberately don't use adjective \textquotedblleft
multivariate" in naming our method. By using a different name we emphasize
the major differences of our method's goals from those of the other
proposals. This also makes it clear that our work is complementary to other
works in this direction. We discuss other connections as we present our main
results.

\subsection{Plan of the paper}

We organize the rest of the paper as follows. In Section 2, we introduce and
develop the properties of CVQF. In Section 3, we introduce and develop the
properties of VQR as well its linear programming implementation. In Section
4, we provide computational details of the discretized form of the linear
programming formulation, which is useful for practice and computation of VQR
with finite samples. In Section 5, we implement VQR in an empirical example,
providing the testing ground for these new concepts. We provide proofs of
all formal results of the paper in the Appendix.

\section{Conditional Vector Quantile Function}

\subsection{Conditional Vector Quantiles as Gradients of Convex Functions}

We consider a random vector $(Y,Z)$ defined on a complete probability space $%
(\Omega _{1},\mathcal{A}_{1},\P _{1})$. The random vector $Y$ takes values
in $\mathbb{R}^{d}$. The random vector $Z$ is a vector covariate, taking
values in $\mathbb{R}^{k}$. Denote by $F_{YZ}$ the joint distribution
function of $(Y,Z)$, by $F_{Y\mid Z}$ the (regular) conditional distribution
function of $Y$ given $Z$, and by $F_{Z}$ the distribution function $Z$. We
also consider random vectors $V$ defined on a complete probability space $%
(\Omega _{0},\mathcal{A}_{0},\P _{0})$, which are required to have a fixed
reference distribution function $F_{U}$. Let $(\Omega ,\mathcal{A},\P )$ be
the a\textit{\ suitably enriched} complete probability space that can carry
all vectors $(Y,Z)$ and $V$ with distributions $F_{YZ}$ and $F_{U}$,
respectively, as well as the independent (from all other variables) standard
uniform random variable on the unit interval. Formally, this product space
takes the form $(\Omega ,\mathcal{A},\P )=(\Omega _{0},\mathcal{A}_{0},\P %
_{0})\times (S_{1},\mathcal{A}_{1},\P _{1})\times ((0,1),B(0,1),\text{Leb})$%
, where $((0,1),B(0,1),\text{Leb})$ is the canonical probability space,
consisting of the unit segment of the real line equipped with Borel sets and
the Lebesgue measure. The symbols $\mathcal{Y}$, $\mathcal{Z}$, $\mathcal{U}$%
, $\mathcal{YZ}$, $\mathcal{UZ}$ denote the support of $F_{Y}$, $F_{Z}$, $%
F_{U}$, $F_{YZ}$, $F_{UZ}$, and $\mathcal{Y}_{z}$ denotes the support of $%
F_{Y\mid Z}(\cdot |z)$. We denote by $\left\Vert .\right\Vert $ the
Euclidian norm of $\mathbb{R}^{d}$.

We assume that the following condition holds:

\begin{itemize}
\item[(N)] \textit{\ $F_U$ has a density $f_U$ with respect to the Lebesgue
measure on $\mathbb{R}^d$ with a convex support set $\mathcal{U}$.}
\end{itemize}

The distribution $F_U$ describes a reference distribution for a vector of
latent variables $U$, taking values in $\mathbb{R}^d$, that we would like to
link to $Y$ via a strong representation of the form mentioned in the
introduction. This vector will be one of many random vectors $V$ having a
distribution function $F_U$, but there will only be one $V=U$, in the sense
specified below, that will provide the required strong representation. The
leading cases for the reference distribution $F_U$ include:

\begin{itemize}
\item the standard uniform distribution on the unit $d$-dimensional cube, $%
U(0,1)^d$,

\item the standard normal distribution $N(0,I_d)$ over $\mathbb{R}^d$, or

\item any other reference distribution on $\mathbb{R}^d$, e.g., uniform on a
ball.
\end{itemize}

Our goal here is to create a deterministic mapping that transforms a random
vector $U$ with distribution $F_{U}$ into $Y$ such that $Y$ conditional on $%
Z $ has the conditional distribution $F_{Y\mid Z}$. Such a map that pushes
forward a probability distribution of interest onto another one is called a
\textit{transport} between these distributions. That is, we want to have a
strong representation property like (\ref{strong0}) that we stated in the
introduction. Moreover, we would like this transform to have \textit{a
monotonicity property}, as in the scalar case. Specifically, in the vector
case we require this transform to be a \textit{gradient of a convex function}%
, which is a plausible generalization of monotonicity from the scalar case.
Indeed, in the scalar case the requirement that the transform is the
gradient of a convex map reduces to the requirement that the transform is
non-decreasing. We shall refer to the resulting transform as the \textit{%
conditional vector quantile function} (CVQF). The following theorem shows
that such map \textit{exists} and is \textit{uniquely determined} by the
stated requirements.

\begin{theorem}[\textbf{CVQF as Conditional Brenier Maps}]
\label{Thm: Conditional Brenier} Suppose condition (N) holds.

(i) There exists a measurable map $(u,z)\longmapsto Q_{Y\mid Z}(u,z)$ from $%
\mathcal{UZ}$ to $\mathbb{R}^{d}$, such that for each $z$ in $\mathcal{Z}$,
the map $u\longmapsto Q_{Y\mid Z}(u,z)$ is the unique ($F_{U}$-almost
everywhere) gradient of convex function such that, whenever $V\sim F_{U}$,
the random vector $Q_{Y\mid Z}(V,z)$ has the distribution function $F_{Y\mid
Z}(\cdot ,z)$, that is,
\begin{equation}
F_{Y\mid Z}(y,z)=\int 1\{Q_{Y\mid Z}(u,z)\leq y\}F_{U}(du),\ \ \text{ for
all }y\in \mathbb{R}^{d}.  \label{df}
\end{equation}%
(ii) Moreover, there exists a random variable $V=U$ such that $\P $-almost
surely
\begin{equation}
Y=Q_{Y\mid Z}(U,Z),\text{ and }\ U\mid Z\sim F_{U}.  \label{strong1}
\end{equation}
\end{theorem}


The theorem is our \textit{first main result} that we announced in the
introduction. It should be noted that the theorem does not require $Y$ to
have an absolutely continuous distribution, it holds for discrete and mixed
outcome variables; only the reference distribution for the latent variable $U
$ is assumed to be absolutely continuous. It is also noteworthy that in the
classical case of $Y$ and $U$ being \textit{scalars} we recover the
classical conditional quantile function as well as the strong representation
formula based on this function (\cite{matzkin}, \cite{Koenker2005}).
Regarding the proof, the first assertion of the theorem is a consequence of
fundamental results due to McCann (\cite{McCann1985}) (as, e.g, stated in
\cite{Villani03}, Theorem 2.32) who in turn refined the fundamental results
of \cite{brenier}. These results were obtained in the case without
conditioning. The second assertion is a consequence of Dudley-Philipp (\cite%
{DP}) result on abstract couplings in Polish spaces.

\begin{remark}[\textbf{Monotonicity}]
The transform $(u,z) \longmapsto (Q_{Y \mid Z} (u, z),z)$ has the following
monotonicity property:
\begin{equation}
(Q_{Y \mid Z}(u,z) - Q_{Y \mid Z}(\bar u,z))^\top (u - \bar u) \geq 0 \ \
\forall u, \bar u \in \mathcal{U}, \forall z \in \mathcal{Z}. \quad \qed
\end{equation}
\end{remark}

\begin{remark}[\textbf{Uniqueness}]
\label{rk:uniqueness}In part (i) of the theorem, $u\longmapsto Q_{Y\mid Z}(u,z)$
is equal to a gradient of some convex function $u\longmapsto \varphi (u,z)$ for $%
F_{U}$-almost every value of $u\in \mathcal{U}$ and it is unique in the
sense that any other map with the same properties will agree with it $F_{U}$%
-almost everywhere. In general, the gradient $u\longmapsto \nabla _{u}\varphi
(u,z)$ exists $F_{U}$-almost everywhere, and the set of points $\mathcal{U}%
_{e}$ where it does not is negligible. Hence the map $u\longmapsto Q_{Y\mid
Z}(u,z)$ is still definable at each $u_{e}\in \mathcal{U}_{e}$ from the
gradient values $\varphi (u,z)$ on $u\in \mathcal{U}\setminus \mathcal{U}%
_{e} $, by defining it at each $u_{e}$ as a smallest-norm element of $\{v\in
\mathbb{R}^{d}:\exists {u_{k}}\in \mathcal{U}\setminus \mathcal{U}%
_{e}:u_{k}\rightarrow u_{e},\nabla _{u}\varphi (u_{k},z)\rightarrow v\}$.
\qed
\end{remark}

\vspace{.1in}

Let us assume further that the following condition holds:

\begin{itemize}
\item[(C)] \textit{\ For each $z \in \mathcal{Z}$, the distribution $F_{Y
\mid Z}(\cdot, z)$ admits a density $f_{Y \mid Z}(\cdot, z)$ with respect to
the Lebesgue measure on $\mathbb{R}^d$.}
\end{itemize}

Under this condition we can recover $U$ uniquely in the following sense:

\begin{theorem}[\textbf{ Conditional Inverse Vector Quantiles or Conditional Vector Ranks}]
\label{Thm: Inverse} Suppose conditions (N) and (C) holds.

Then there exists a measurable map $(y,z) \longmapsto Q^{-1}_{Y\mid Z}(y,z)$,
mapping $\mathcal{YZ}$ to $\mathbb{R}^d$, such that for each $z$ in $%
\mathcal{Z}$, the map $y \longmapsto Q^{-1}_{Y\mid Z}(y,z)$ is the inverse of $u
\longmapsto Q_{Y\mid Z}(u,z)$ in the sense that:
\begin{equation*}
Q^{-1}_{Y\mid Z} ( Q_{Y\mid Z}(u,z),z) = u,
\end{equation*}
for almost all $u$ under $F_U$. Furthermore, we can construct $U$ in (\ref%
{strong1}) as follows,
\begin{equation}  \label{strong2}
U = Q^{-1}_{Y\mid Z}(Y,Z), \text{ and } \ U\mid Z \sim F_U.
\end{equation}
\end{theorem}

\begin{remark}[Conditional Vector  Rank Function]
The mapping  $y \longmapsto Q_{Y\mid Z}^{-1}(y,z)$, which maps $\mathcal{Y}
\subset \mathbb{R}^d$  to $\mathbb{R}^d$, is the conditional rank function.
When $d=1$, it coincides with the conditional distribution function, but
when $d>1$ it does not. The ranking interpretation stems from the fact that
when we set $F_{U} = U(0,1)^d$, vector $Q_{Y|Z}^{-1}\left( Y,Z\right) \in %
\left[ 0,1\right] ^{d}$ measures the centrality of observation $Y$ for each
of the dimensions, conditional on $Z$. \qed
\end{remark}

\bigskip

It is also of interest to state a further implication, which occurs under
(N) and (C), on the link between the transportation map $Q_{Y\mid Z}$ and
its derivatives on one side, and the densities $f_{U}$ and $f_{Y|Z}$ on the
other side. This link is a nonlinear second order partial differential
equation called a (conditional) \emph{Monge-Amp\`{e}re equation}.

\begin{corollary}[\textbf{Conditional Monge-Amp\`ere Equations}]
Assume that conditions (N) and (C) hold and, further, that the map $u
\longmapsto Q_{Y \mid Z} (u,z)$ is continuously differentiable and injective for
each $z \in \mathcal{Z}$. Under this condition, the following \textit{%
conditional forward Monge-Amp\`ere equation} holds for all $(u,z) \in
\mathcal{UZ}$ : {\small
\begin{equation}  \label{monge-ampere1}
f_U(u) = f_{Y \mid Z} (Q_{Y \mid Z} (u,z),z) \text{det} [\mathrm{D}_u Q_{Y
\mid Z} (u,z)] = \int \delta( u- Q^{-1}_{Y \mid Z}(y,z)) f_{Y \mid Z}(y,z)
dy,
\end{equation}%
}\!\! where $\delta$ is the Dirac delta function in $\mathbb{R}^d$ and $%
\mathrm{D}_u= \partial/\partial u^\top$. Reversing the roles of $U$ and $Y$,
we also have the following \textit{conditional backward Monge-Amp\`ere
equation} holds for all $(y,z) \in \mathcal{YZ}$: {\small
\begin{equation}  \label{monge-ampere2}
f_{Y\mid Z}(y,z) = f_U(Q^{-1}_{Y \mid Z} (y,z)) \text{det} [\mathrm{D}_y
Q^{-1}_{Y \mid Z} (y,z)] = \int \delta( y- Q_{Y \mid Z}(u,z)) f_{U}(u) du.
\end{equation}%
}
\end{corollary}


The latter expression is useful for linking the conditional density function
to the conditional vector quantile function. Equations (\ref{monge-ampere1})
and (\ref{monge-ampere2}) are \textit{partial differential equations} of the
Monge-Amp\`ere type, carrying an additional index $z \in \mathcal{Z}$. These
equations could be used directly to solve for conditional vector quantiles
given conditional densities. We can also use them to set up maximum
likelihood  method for recovering conditional vector quantiles.
In the next section we describe a variational
approach to recovering conditional vector quantiles.

\subsection{Conditional Vector Quantiles as Optimal Transport}

Under additional moment assumptions, the CVQF can be characterized and even
defined as solutions to a regression version of the
Monge-Kantorovich-Brenier's optimal transportation problem or, equivalently,
a conditional correlation maximization problem.

We assume that the following conditions hold:

\begin{itemize}
\item[(M)] The second moment of $Y$ and the second moment of $U$ are finite:%
\begin{equation*}
\int \int \|y\|^2 F_{YZ}(dy, d z) < \infty \text{ and } \int \|u\|^2 F_U(du)
< \infty.
\end{equation*}
\end{itemize}

We consider the following optimal transportation problem with conditional
independence constraints:
\begin{equation}  \label{transport1}
\min_V \{ {\mathrm{E}} \|Y-V\|^2 : V \mid Z \sim F_U \},
\end{equation}
where the minimum is taken over all random vectors $V$ defined on the
probability space $(\Omega, \mathcal{F}, \P )$. Note that the value of
objective is the Wasserstein distance between $Y$ and $V$ subject to $V \mid
Z\sim F_U$. Under condition (M) we will see that a \textit{solution} exists
and is given by $V=U$ constructed in the previous section.

The problem (\ref{transport1}) is the conditional version of the classical
Monge - Kantorovich problem with Brenier's quadratic costs, which was solved
by Brenier in considerable generality in the unconditional case. In the
unconditional case, the canonical Monge problem is to transport a pile of
coal with mass distributed across production locations from $F_{U}$ into a
pile of coal with mass distributed across consumption locations from $F_{Y}$%
, and it can be rewritten in terms of random variables $V$ and $Y$. We are
seeking to match $Y$ with a version of $V$ that is closest in mean squared
sense subject to $V$ having a prescribed distribution. Our conditional
version above (\ref{transport1}) imposes the additional conditional
independence constraint $V\mid Z\sim F_{U}$.

The problem above is equivalent to covariance maximization problem subject
to the prescribed conditional independence and distribution constraints:
\begin{equation}
\max_{V}\{{\mathrm{E}}(V^{\top }Y):\ \ V\mid Z\sim F_{U}\},  \label{maxcorr}
\end{equation}%
where the maximum is taken over all random vectors $V$ defined on the
probability space $(\Omega ,\mathcal{F},\P )$. This type of problem will be
convenient for us, as it most directly connects to convex analysis and leads
to a convenient dual program. This form also connects to unconditional
multivariate quantile maps defined in \cite{egh}, who employed them for
purposes of risk analysis; our definition given in the previous section is
more satisfactory, because it does not require any moment conditions, as
follows from the results of \cite{McCann1985}.


The dual program to (\ref{maxcorr}) can be stated as:
\begin{equation}  \label{dualmaxcorr}
\begin{array}{lll}
\min_{(\psi, \varphi)} {\mathrm{E}}(\varphi(V,Z) + \psi(Y,Z)) & :
\;\varphi(u, z) + \psi(y,z) \ge u^\top y &  \\
& \ \text{ for all} \ (z, y, u) \in \mathcal{Z}\times\mathbb{R}^{2d}, &
\end{array}%
\end{equation}
where $V$ is any vector such that $V \mid Z \sim F_U$, and minimization is
performed over Borel maps $(y,z) \longmapsto \psi(y,z)$ from $\mathcal{Z}\times%
\mathbb{R}^{d}$ to $\mathbb{R} \cup \{+ \infty\}$ and $(u,z) \longmapsto
\varphi(z,u)$ from $\mathcal{Z}\times\mathbb{R}^{d}$ to $\mathbb{R} \cup \{+
\infty\}$, where $y \longmapsto \psi(y,z)$ and $u \longmapsto \varphi(u,z)$ are
lower-semicontinuous for each value $z \in \mathcal{Z}$.

\begin{theorem}[\textbf{Conditional Vector Quantiles as Optimal Transport}]
\label{thm: variational characterization} Suppose conditions (N), (C), and
(M) hold.

(i) There exists a pair of maps $(u,z) \longmapsto \varphi (u,z) $ and $(y,z)
\longmapsto \psi(y,z)=\varphi^*(y,z)$, each mapping from $\mathbb{R}^{d}\times
\mathcal{Z}$ to $\mathbb{R}$, that solve the problem (\ref{dualmaxcorr}).
For each $z \in \mathcal{Z}$, the maps $u \longmapsto \varphi (u,z) $ and $y
\longmapsto \varphi^*(y,z)$ are convex and are Legendre transforms of each
other:
\begin{equation*}
\varphi(u,z) = \sup_{y \in \mathbb{R}^d} \{ u^\top y - \varphi^*(y,z) \},
\quad \varphi^*(y,z) = \sup_{u \in \mathbb{R}^d } \{ u^\top y - \varphi(u,z)
\},
\end{equation*}
for all $(u,z) \in \mathcal{UZ}$ and $(y,z) \in \mathcal{YZ}$.

(iii) We can take the gradient $(u,z) \longmapsto \nabla_u \varphi(u,z)$ of $%
(u,z) \longmapsto \varphi(u,z)$ as the conditional vector quantile function,
namely, for each $z \in \mathcal{Z}$, $Q_{Y\mid Z}(u,z) = \nabla_u
\varphi(u,z)$ for almost every value $u$ under $F_U$.

(iv) We can take the gradient $(y,z) \longmapsto \nabla_y \varphi^*(y,z)$ of $%
(y,z) \longmapsto \varphi^*(y,z)$ as the conditional inverse vector quantile
function or conditional vector rank function, namely, for each $z \in \mathcal{Z}$,
$Q^{-1}_{Y\mid Z}(y,z) = \nabla_y \varphi^*(z,y)$ for almost every value $y$
under $F_{Y \mid Z}(\cdot,z)$.

(v) The vector $U = Q^{-1}_{Y\mid Z}(Y,Z)$ is a solution to the primal
problem (\ref{maxcorr}) and is unique in the sense that any other solution $%
U^*$ obeys $U^* = U$ almost surely under $\P $. The primal (\ref{maxcorr})
and dual (\ref{dualmaxcorr}) have the same value.

(vi) The maps $u\longmapsto \nabla _{u}\varphi (u,z)$ and $y\longmapsto \nabla
\varphi _{y}^{\ast }(y,z)$ are inverses of each other: for each $z\in
\mathcal{Z}$, and for almost every $u$ under $F_{U}$ and almost every $y$
under $F_{Y\mid Z}(\cdot ,z)$
\begin{equation*}
\nabla _{y}\varphi ^{\ast }(\nabla _{u}\varphi (u,z),z)=u,\quad \nabla
_{u}\varphi (\nabla _{y}\varphi ^{\ast }(y,z),z)=y.
\end{equation*}
\end{theorem}

\bigskip 

\begin{remark}
\label{rk:optimality}There are many maps $Q:\mathcal{UZ}\rightarrow \mathcal{%
Y}$ such that if $V\sim F_{U}$, then $Q\left( V,z\right) \sim F_{Y|Z=z}$.
Any of these maps define a \emph{transport} from $F_{U}$ to $F_{Y|Z=z}$. Our
choice is to take the \emph{optimal transport}, in the sense that it
minimizes the Wasserstein distance $\mathrm{E} \| Q\left( V,Z\right) - V \|^2
$ among such maps. This has several benefits: (i) the optimal transport is
unique as soon as $F_{U}$ is absolutely continuous, as noted in Remark \ref%
{rk:uniqueness} and (ii) this object is easily computable through a linear
programming problem. Note that the classical, scalar quantile map is the
optimal transport from $F_{U}$ to $F_{Y}$ in this sense, so oue notion
indeed extends the classical notion of a quantile. \qed
\end{remark}

\begin{remark}
\label{rk:locality}Unlike in the scalar case, we cannot compute $
Q_{Y|Z}\left( u,z\right)$ at a given point $u$ without computing the whole
map $u\rightarrow Q_{Y|Z}\left( u,z\right) $. This highlights the fact that
CVQF is not a local concept with respect to values of the rank $u$. \qed \end{remark}

Theorem~\ref{thm: variational characterization} provides a number of
analytical properties, formalizing the variational interpretation of
conditional vector quantiles, providing the potential functions $
(u,z)\longmapsto \varphi (u,z)$ and $(y,z)\longmapsto \varphi ^{\ast }(y,z)$, which
are mutual Legendre transforms, and whose gradients are the conditional
vector quantile functions and its inverse, the conditional vector rank function. This problem is a conditional
generalization of the fundamental results by Brenier as presented in \cite%
{Villani03}, Theorem 2.12.

\begin{example}[\textbf{Conditional Normal Vector Quantiles}]
Here we consider the normal conditional vector quantiles. Consider the case
where
\begin{equation*}
Y\mid Z\sim N(\mu (Z),\Omega (Z)).
\end{equation*}%
Here $z\longmapsto \mu (z)$ is the conditional mean function and $z\longmapsto
\Omega (z)$ is a conditional variance function such that $\Omega (z)>0$ (in
the sense of positive definite matrices) for each $z\in \mathcal{Z}$ with ${%
\mathrm{E}}\Vert \Omega (Z)\Vert +{\mathrm{E}}\Vert \mu (Z)\Vert ^{2}<\infty
$. The reference distribution is given by $U\mid Z\sim N(0,I)$. Then we have
the following conditional vector quantile model:
\begin{eqnarray*}
&&Y=\mu (Z)+\Omega ^{1/2}(Z)U, \\
&&U=\Omega ^{-1/2}(Z)(Y-\mu (Z)).
\end{eqnarray*}%
Here we have the following conditional potential functions
\begin{eqnarray*}
&&\varphi (u,z)=\mu (z)^{\top }u+\frac{1}{2}u^{\top }\Omega ^{1/2}(z)u, \\
&&\psi (y,z)=\frac{1}{2}(y-\mu (z))^{\top }\Omega ^{-1/2}(z)(y-\mu (z)),
\end{eqnarray*}%
and the following conditional vector quantile and rank functions:
\begin{eqnarray*}
&&Q_{Y\mid Z}(u,z)=\nabla _{u}\varphi (u,z)=\mu (z)+\Omega ^{1/2}(z)u, \\
&&Q_{Y\mid Z}^{-1}(y,z)=\nabla _{y}\psi (y,z)=\Omega ^{-1/2}(z)(y-\mu (z)).
\end{eqnarray*}%
It follows from Theorem~\ref{thm: variational characterization} that $V=U$
solves the covariance maximization problem (\ref{maxcorr}). This example is
special in the sense that the conditional vector quantile and rank functions
are \textit{linear in }$u$ and $y$, respectively. \qed
\end{example}

\subsection{Interpretations of vector rank $U$}

\label{rk:latentFactor} We can provide the following interpretations of $U$:

\bigskip

\textbf{1) As multivariate rank}. An interesting interpretation of $U$ is as
a multivariate rank. In the univariate case, \cite{Koenker2005}, Ch. 1.3 and
3.5, interprets $U$ as a continuous notion of rank in the setting of
quantile regression. The rank has a reference distribution $F_{U}$, which is
typically chosen to be uniform on $(0,1)$, but other reference distributions
could be used as well. The concept of vector quantile allows us to assign a
continuous rank to each of the dimensions, and the vector quantile mapping
is monotone with respect to the rank in the sense of being the gradient of a
convex function. As a result, $U$ can be interpreted as a \emph{multivariate
rank} for $Y$, as we are trying to map the distribution of $U$ to a
prescribed distribution $F_{Y}$ at minimal distortion, as seen in (\ref%
{transport1}). \newline

\textbf{2) As a reference outcome for defining quantile treatment effects}.
Another motivation is related to the classical definition of quantile
treatment effects introduced by \cite{lehman}, and further developed by \cite%
{doksum}, \cite{Koenker2005}, and others. Suppose we define $U$ as an
outcome for an untreated population; for this we simply set the reference
distribution $F_{U}$ to the distribution of outcome in the untreated
population. Suppose $Z$ is the indicator of the receiving a treatment ($Z=0$
means no treatment). Then we can represent outcome $Y=Q_{Y\mid Z}\left(
U,Z\right) $ as the multivariate health outcome conditional on $Z$. If $Z=0$%
, then the outcome is distributed as $Q_{Y\mid Z}\left( U,0\right) =U$. If $%
Z=1$, then the outcome is distributed as $Q_{Y\mid Z}\left( U,1\right) $.
The corresponding notion of vector quantile treatment effects is $Q_{Y\mid
Z}(u,1)-Q_{Y\mid Z}(u,0).$

\bigskip

\textbf{3) As nonlinear latent factors}. As it is apparent in the
variational formulation~(\ref{transport1}), the entries of $U$ can also be
thought as latent factors, independent of each other and explanatory
variables $Z$ and having a prescribed marginal distribution $F_{U}$, and
that best explain the variation in $Y$. Therefore, the conditional vector
quantile model (\ref{strong1}) provides a non-linear latent factor model for
$Y$ with factors $U$ solving the matching problem (\ref{transport1}). This
interpretation suggests that this model may be useful in applications which
require measurement of multidimensional unobserved factors, for example,
cognitive ability, persistence, and various other latent propensities; see,
for example, \cite{JJH}.

\subsection{Overview of Other Notions of Multivariate Quantile\label%
{par:discussionOthers}}

We briefly review other notions of multivariate quantiles in the statistical
literature. We highlight the main contrasts with the notion we are using,
based on optimal transport. For the sake of clarity of exposition, we
discuss the unconditional case; albeit the comparisons extend naturally to
the regression case.

In \cite{chaudhuri}, the following definition of multivariate quantile
function is suggested: for $u\in \mathbb{R}^{d}$, let%
\begin{equation*}
Q_{Y}^{C}\left( u\right) =\arg \max_{y\in \mathbb{R}^{d}}{\mathrm{E}}\left[
y^{\top }u-\Vert y-Y\Vert \right]
\end{equation*}%
which coincides with the classical notion when $d=1$. See also \cite%
{serfling}. More generally, \cite{Koltchinskii} offers the following
definition based on M-estimators, still for $u\in \mathbb{R}^{d}$,%
\begin{equation*}
Q_{Y}^{K}\left( u\right) =\arg \max_{y\in \mathbb{R}^{d}}{\mathrm{E}}\left[
y^{\top }u-K\left( y,Y\right) \right]
\end{equation*}%
for a choice of kernel $K$ assumed to be convex with respect to its first
argument. Like our proposal, these notions of quantile maps are gradients of
convex potentials. However, unlike our proposal, these notions do not
provide a transport from a fixed distribution over values of $u$ to the
distribution $F_{Y}$ of $Y$ as soon as $d>1$.

In \cite{Wei}, a notion of quantile based on the Rosenblatt map is investigated. In the case $%
d=2$, this quantile is defined for $u\in \left[ 0,1\right] ^{2}$ as%
\begin{equation*}
Q_{Y}^{R}\left( u_{1},u_{2}\right) =\left( Q_{Y_{1}}\left( u_{1}\right)
,Q_{Y_{2}\mid Y_{1}}\left( u_{2}\mid Q_{Y_{1}}\left( u_{1}\right) \right)
\right)
\end{equation*}%
where $Q_{Y_{1}}$ and $Q_{Y_{1}\mid Y_{2}}$ are the univariate and the
conditional univariate quantile map. This map is a transport of the
distribution of $\mathcal{U}(0,1)^{2}$; however, in this definition, $Y_{1}$
and $Y_{2}$ play sharply assymetric roles, as the second dimension is
defined conditional on the first one. Unlike ours, this quantile map is not
a gradient of convex function.

In \cite{HPS}, the authors specify a vector of latent indices $u\in {B}^{d}$ the unit ball
of $\mathbb{R}^{d}$. For $u\in {B}^{d}$, they define multivariate quantiles
as
\begin{equation*}
Q_{Y}^{HPS}\left( u\right) =\left\{ y\in \mathbb{R}^{d}:c^{\top }y=a\right\}
,
\end{equation*}%
where $a\in \mathbb{R}$ and $c\in \mathbb{R}^{d}$ minimize $\mathrm{E}\rho
_{\Vert u\Vert }\left( c^{\top }Y-a\right) $ subject to constraint $c^{\top
}u=\Vert u\Vert $. In contrast to ours, their notion of quantile is a
set-valued. A closely related construction is provided by \cite{KongMizera}
who define the directional quantile associated to the index $u\in \mathbf{B}%
^{d}$ via:
\begin{equation*}
Q_{Y}^{KM}\left( u\right) =Q_{u^{\top }Y/\Vert u\Vert }\left( \Vert u\Vert
\right) u/\Vert u\Vert ,
\end{equation*}%
where $Q_{u^{\top }Y/\Vert u\Vert }$ is the univariate quantile function of
the random variable $u^{\top }Y/\Vert u\Vert $. We can provide a transport
interpretation to this notion of quantiles, but unlike our proposal this map
is not a gradient of convex function.


A notion of quantile based on a partial
order $\preceq $ on $\mathbb{R}^{d}$ is proposed in \cite{Belloni:Wrinkler}. For an index $u\in (0,1)$, these authors define%
\begin{equation*}
Q_{Y}^{BW}\left( u\right) =\left\{ y\in \mathbb{R}^{d}:\Pr \left( Y\succeq
y\mid C\left( y\right) \right) \geq 1-u,\Pr \left( Y\preceq y\mid C\left(
y\right) \right) \geq u\right\}
\end{equation*}%
where $C\left( y\right) =\left\{ y^{\prime }\in \mathbb{R}^{d}:y\succeq
y^{\prime }\text{ or }y^{\prime }\succeq y\right\} $ is the set of elements
that can be ordered by $\succeq $ relative to the point $y$. Unlike our
proposal, the index $u$ is scalar and the quantile is set-valued.

\section{Vector Quantile Regression}

\subsection{Linear Formulation}

Here we let $X=f(Z)$ denote a vector of regressors formed as transformations
of $Z$, such that the first component of $X$ is 1 (intercept term in the
model) and such that conditioning on $X$ is equivalent to conditioning on $Z$%
. The dimension of $X$ is denoted by $p$ and we shall denote $X=(1,X_{-1}^\top)^\top$
with $X_{-1}\in {\mathbb{R}}^{p-1}$.

In practice, $X$ would often consist of a constant and some polynomial or
spline transformations of $Z$ as well as their interactions. Note that
conditioning on $X$ is equivalent to conditioning on $Z$ if, for example, a
component of $X$ contains a one-to-one transform of $Z$.

Denote by $F_X$ the distribution function of $X$ and $F_{UX} = F_U F_X$. Let
$\mathcal{X}$ denote the support of $F_X$ and $\mathcal{UX}$ the support of $%
F_{UX}$. We define \textit{linear vector quantile regression model} (VQRM)
as the following linear model of CVQF.

\begin{itemize}
\item[(L)] The following linearity condition holds:
\begin{equation*}
Y=Q_{Y\mid X}(U,X)=\beta _{0}(U)^{\top }X,\ \ U\mid X\sim F_{U},
\end{equation*}%
where $u\longmapsto \beta _{0}(u)$ is a map from $\mathcal{U}$ to the set $%
\mathcal{M}_{p\times d}$ of $p\times d$ matrices such that $u\longmapsto \beta
_{0}(u)^{\top }x$ is a monotone, smooth map, in the sense of being a
gradient of a convex function:
\begin{equation*}
\beta _{0}(u)^{\top }x=\nabla _{u}\Phi _{x}(u),\ \ \Phi
_{x}(u):=B_{0}(u)^{\top }x,\text{ for all }(u,x)\in \mathcal{UX},
\end{equation*}%
where $u\longmapsto B_{0}(u)$ is $C^{1}$ map from $\mathcal{U}$ to $\mathbb{R}%
^{d}$, and $u\longmapsto B_{0}(u)^{\top }x$ is a strictly convex map from $%
\mathcal{U}$ to $\mathbb{R}$.
\end{itemize}

The parameter $\beta (u)$ is indexed by the quantile index $u\in \mathcal{U}$
and is a $d\times p$ \textit{matrix} of quantile regression coefficients. Of
course in the scalar case, when $d=1$, this matrix reduces to a vector of
quantile regression coefficients. This model is a natural analog of the
classical QR for scalar $Y$ where the similar regression representation
holds. One example where condition (L) holds is Example 2.1, describing the
conditional normal vector regression. It is of interest to specify other
examples where condition (L) holds or provides a plausible approximation.

\begin{example}[Saturated Specification]
\label{example:S} The regressors $X=f(Z)$ with ${\mathrm{E}} \|f(Z)\|^2 <
\infty$ are \textit{saturated} with respect to $Z$, if, for any $g \in
L^2(F_Z)$ , we have $g(Z) = X^\top \alpha_g$.
In this case the linear functional form (L) is not a restriction. For $%
p<\infty$ this can occur if and only if $Z$ takes on a finite set of values $%
\mathcal{Z} = \{z_1,\dots, z_p\}$, in which case we can write:
\begin{equation*}
Q_{Y|X}(u,X) = \sum_{j=1}^p Q_{Y \mid Z}(u,z_j) 1(Z=z_j) =: B_0(u)^\top X,
\end{equation*}%
\begin{equation*}
B_0(u) := \left(
\begin{array}{c}
Q_{Y \mid Z}(u,z_1)^\top \\
\vdots \\
Q_{Y \mid Z}(u,z_p)^\top \\
\end{array}
\right), \quad X:= \left(
\begin{array}{c}
1(Z=z_1) \\
\vdots \\
1(Z=z_{p}) \\
\end{array}
\right).
\end{equation*}
Here the problem is equivalent to considering $p$ unconditional vector
quantiles in populations corresponding to $Z=z_1, \dots, Z=z_p$. \qed
\end{example}

The rationale for using linear forms is two-fold -- one is convenience of
estimation and representation of functions and another one is approximation
property. We can approximate a smooth convex potential by a smooth linear
potential, as the following example illustrates for a particular
approximation method.

\begin{example}[Linear Approximation]
Let $(u,z) \longmapsto \varphi(u,z)$ be of class $C^a$ with $a>1$ on the support
$(u,z) \in \mathcal{UZ} = [0,1]^{d+k}$. Consider a trigonometric tensor
product basis of functions $\{(u,z) \longmapsto q_j(u) f_l(z), j \in\mathbb{N},
l \in \mathbb{N}\}$ in $L^2[0,1]^{d+k}$. Then there exists a $JL$ vector $%
(\gamma_{jl}: j \in \{1,...,J\},l\in \{1,...,L\})$ such that the linear map:
\begin{equation*}
(u,z) \longmapsto \Phi^{JL}(u,z) := \sum_{j =1}^J \sum_{l =1}^L \gamma_{jl}
q_j(u) f_l(z) =: B^L_0(u)^\top f^L(z),
\end{equation*}
where $B^L_0(u) = ( \sum_{j =1}^J \gamma_{jl} q_j(u), l\in \{1,...,L\})$ and
$f^L(z)= ( f_l(z), l\in \{1,...,L\})$, provides uniformly consistent
approximation of the potential and its derivative:
\begin{equation*}
\lim_{J,L \to \infty }\sup_{(u,z) \in \mathcal{UZ}} \Big ( | \varphi(u,z) -
\Phi^{JL} (u,z) | + \| \nabla_u \varphi(u,z) - \nabla_u \Phi^{JL} (u,z) \| %
\Big ) = 0. \quad \qed
\end{equation*}
\end{example}

The approximation property via the sieve-type approach provides a rationale for the linear (in parameters) specification
(\ref{eq: linear factor}).   Another approach, based on local polynomial approximations over a collection
of (increasingly smaller) neighborhoods, also provides a useful rationale for the linear (in parameters) specification,
e.g., similarly in spirit to  \cite{YuJones}.   If the linear specification does not hold \textit{exactly} we say
that the model is \textit{misspecified}. If the model is flexible enough, by using a suitable
basis or localization,
then the approximation error is small, and we effectively ignore the error
when assuming (\ref{eq: linear factor}). However, when constructing a
sensible estimator we must allow the possibility that the model is
misspecified, which means we can't really force (\ref{eq: linear factor})
onto data. Our proposal for estimation presented next does not force (\ref%
{eq: linear factor}) onto data, but if (\ref{eq: linear factor}) is true in
population, then as a result, the true conditional vector quantile function
would be recovered perfectly in population.

\subsection{Linear Program for VQR}

Our approach to multivariate quantile regression is based on the
multivariate extension of the covariance maximization problem with a mean
independence constraint:
\begin{equation}  \label{maxcorrmimd}
\max_V \{ {\mathrm{E}}(V^\top Y): \; V\sim F_U, \; {\mathrm{E}}(X \mid V)={%
\mathrm{E}}(X) \}.
\end{equation}

Note that the constraint condition is a relaxed form of the previous
independence condition.

\begin{remark}
The new condition $V\sim F_U, \; {\mathrm{E}}(X\mid V)={\mathrm{E}}(X)$ is
weaker than $V \mid X \sim F_U$, but the two conditions coincide if $X$ is
saturated relative to $Z$, as in Example \ref{example:S}, in which case ${%
\mathrm{E}} (g(Z) V) = {\mathrm{E}} X^{\top }\alpha_g V = {\mathrm{E}}%
(X^{\top }\alpha_g) {\mathrm{E}}( V) = {\mathrm{E}} g(Z) {\mathrm{E}} V$ for
every $g \in L^2(F_Z)$. More generally, this example suggests that the
richer $X$ is, the closer the mean independence condition becomes to the
conditional independence. \qed
\end{remark}

The relaxed condition is sufficient to guarantee that the solution exists
not only when (L) holds, but more generally when the following quasi-linear
assumption holds.

\begin{itemize}
\item[(QL)] We have a quasi-linear representation a.s.
\begin{equation*}
Y=\beta (\widetilde{U})^{\top }X,\ \ \widetilde{U}\sim F_{U},\ \ {\mathrm{E}}%
(X\mid \widetilde{U})={\mathrm{E}}(X),
\end{equation*}%
where $u\longmapsto \beta (u)$ is a map from $\mathcal{U}$ to the set $\mathcal{M%
}_{p\times d}$ of $p\times d$ matrices such that $u\longmapsto \beta (u)^{\top
}x $ is a gradient of convex function for each $x\in \mathcal{X}$ and a.e. $%
u\in \mathcal{U}$:
\begin{equation*}
\beta (u)^{\top }x=\nabla _{u}\Phi _{x}(u),\ \ \Phi _{x}(u):=B(u)^{\top }x,
\end{equation*}%
where $u\longmapsto B(u)$ is $C^{1}$ map from $\mathcal{U}$ to $\mathbb{R}^{d}$,
and $u\longmapsto B(u)^{\top }x$ is a strictly convex map from $\mathcal{U}$ to $%
\mathbb{R}$.
\end{itemize}

This condition allows for a degree of misspecification, which allows for a
latent factor representation where the latent factor obeys the relaxed
independence constraints.

\begin{theorem}
\label{theorem: linear} Suppose conditions (M), (N) , (C), and (QL) hold.

(i) The random vector $\widetilde U$ entering the quasi-linear
representation (QL) solves (\ref{maxcorrmimd}).

(ii) The quasi-linear representation is unique a.s. that is if we also have $%
Y= \overline{\beta}(\overline{U})^\top X$ with $\overline{U} \sim F_U, {%
\mathrm{E}} (X \mid \overline{U}) = {\mathrm{E}} X$, $u \longmapsto X^\top
\overline{\beta}(u)$ is a gradient of a strictly convex function in $u \in
\mathcal{U}$ a.s., then $\bar U= \tilde U$ and $X^\top \beta(\tilde
U)=X^\top \overline{\beta}(\tilde U)$ a.s.

(iii) Under condition (L) and assuming that ${\mathrm{E}} (X X^\top)$ has
full rank, $\widetilde U =U$ a.s. and $U$ solves (\ref{maxcorrmimd}).
Moreover, $\beta_0(U) = \beta(U)$ a.s.
\end{theorem}

The last assertion is important -- it says that if (L) holds, then the
linear program~(\ref{maxcorrmimd}), where the independence constraint has
been relaxed into a \emph{mean independence} constraint, will find the true
linear vector quantile regression in the population.

\subsection{Dual Program for Linear VQR}

As explained in details in the appendix, Program~(\ref{maxcorrmimd}) is an
infinite-dimensional linear programming problem whose dual program is:
\begin{equation}
\begin{array}{ll}
\inf_{(\psi ,b)} \ {\mathrm{E}} \psi (X,Y)+{\mathrm{E}}b(V)^{\top }{\mathrm{%
E}}(X)\;:\; & \psi (x,y)+b(u)^{\top }x\geq u^{\top }y, \\
& \forall \ (y,x,u)\in \mathcal{Y}\mathcal{X}\mathcal{U},%
\end{array}
\label{dualmimc}
\end{equation}%
where $V\sim F_{U}$, where the infimum is taken over all continuous
functions $(y,x)\longmapsto \psi (y,x)$, mapping $\mathcal{YX}$ to $\mathbb{R}$
and $u\longmapsto b(u)$ mapping $\mathcal{U}$ to $\mathbb{R}$, such that ${%
\mathrm{E}} \psi (X,Y)$ and ${\mathrm{E}}b(V)$ are finite.

Since for fixed $b$, the smallest $\psi $ which satisfies the pointwise
constraint in (\ref{dualmimc}) is given by
\begin{equation*}
\psi (x,y):=\sup_{u\in \mathcal{U}}\{u^{\top }y-b(u)^{\top }x\},
\end{equation*}%
one may equivalently rewrite (\ref{dualmimc}) as the minimization over
continuous $b$ of
\begin{equation*}
\int \sup_{u\in \mathcal{U}}\{u^{\top }y-b(u)^{\top }x\}F_{Y\mid
X}(dx,dy)+\int b(u)^{\top }\mathrm{E}(X)F_{U}(du).
\end{equation*}%
By standard arguments (\cite{Villani03}, section 1.1.7), the
infimum over continuous functions coincides with the one over smooth or
simply integrable functions.

\begin{theorem}
\label{theorem: dual1} Under (M) and (QL), we have that the optimal solution
to the dual is given by functions:
\begin{equation*}
\psi(x,y)=\sup_{u \in \mathcal{U}} \{ u^\top y- B(u)^\top x\}, \ \ b(u) =
B(u).
\end{equation*}
\end{theorem}

This result can be recognized as a consequence of strong duality of the
linear programming (e.g. \cite{Villani03}).


\subsection{\label{sec:scalar}Connecting to Scalar Quantile Regression}

We now consider the connection to the canonical, scalar quantile regression
primal problem, where $Y$ is scalar and for each probability index $t \in (0,1)$, the
linear functional form $x\longmapsto x^{\top }\beta (t)$ is used. \cite{KB}
define linear quantile regression as $X^{\top }\beta (t)$ with $\beta (t)$
solving the minimization problem
\begin{equation}
\beta (t)\in \arg \min_{\beta \in {\mathbb{R}}^{p}}{\mathrm{E}} \rho
_{t}(Y-X^{\top }\beta ),  \label{kb0}
\end{equation}%
where $\rho
_{t}(z):=tz_{-}+(1-t)z_{+}$, with $z_{-}$ and $z_{+}$ denoting the
negative and positive parts of $z$. The above formulation makes sense and $%
\beta (t)$ is unique under the following simplified conditions:

\begin{itemize}
\item[(QR)] ${\mathrm{E}}|Y| < \infty$ and $\E \|X\|^2 < \infty$, $(y,x) \longmapsto f_{Y\mid X}(y,x)$ is
bounded and uniformly continuous, and ${\mathrm{E}} ( f_{Y \mid X}(X^{\top }\beta(t),X) X X^\top)$ is of full rank.
\end{itemize}

We note that (\ref{kb0}) can be conveniently rewritten as
\begin{equation}
\min_{\beta \in {\mathbb{R}}^{p}}\{{\mathrm{E}}(Y-X^{\top }\beta )_{+}+(1-t){%
\mathrm{E}}X^{\top }\beta \}.  \label{kb1}
\end{equation}%
\cite{KB} showed that this convex program admits as \textit{dual formulation}%
:
\begin{equation}
\max \{{\mathrm{E}}(A_{t}Y)\;:\;A_{t}\in \lbrack 0,1],\;{\mathrm{E}}%
(A_{t}X)=(1-t){\mathrm{E}}X\}.  \label{dt}
\end{equation}%
An optimal $\beta =\beta (t)$ for (\ref{kb1}) and an optimal rank-score
variable $A_{t}$ in (\ref{dt}) may be taken to be
\begin{equation}
A_{t}=1(Y>X^{\top }\beta (t)),  \label{frombetatoU}
\end{equation}%
and thus the constraint ${\mathrm{E}}(A_{t}X)=(1-t){\mathrm{E}}X$ reads:
\begin{equation}
{\mathrm{E}}(1(Y>X^{\top }\beta (t))X)=(1-t){\mathrm{E}}X, \label{normaleq}
\end{equation}%
which simply is the first-order optimality condition for (\ref{kb1}).

We say that the specification of quantile regression is quasi-linear if
\begin{equation}  \label{monqrqs}
t\longmapsto x^\top \beta(t) \mbox{ is increasing on $(0,1)$}.
\end{equation}
Define the rank variable $\tilde U=\int_0^1 A_t dt$, then under (\ref{monqrqs}) we have that
\begin{equation*}
A_t = 1(\tilde U>t),
\end{equation*}
and the first-order condition (\ref{normaleq}) implies that for each $t \in (0,1)$
\begin{equation*}
{\mathrm{E}} 1 ( \tilde U \geq t) = (1-t), \ \ {\mathrm{E}}1 ( \tilde U \geq
t) X = (1-t) {\mathrm{E}} X.
\end{equation*}
The first property implies that $\tilde U \sim U(0,1)$ and the second
property can be easily shown to imply the \textit{mean-independence}
condition:
\begin{equation*}
{\mathrm{E}}(X \mid \tilde U) = {\mathrm{E}} X .
\end{equation*}
Thus quantile regression \textit{naturally leads} to the mean-independence
condition and the quasi-linear latent factor model. This is the reason we
used mean-independence condition as a starting point in formulating the
vector quantile regression. Moreover, in both vector and scalar cases, we
have that, when the conditional quantile function is linear (not just
quasi-linear), the quasi-linear representation coincides with the linear
representation and $\tilde U$ becomes fully independent of $X$.

The following result summarizes the connection more formally.

\begin{theorem}[\textbf{Connection to Scalar QR}]
\label{qrqsdec} Suppose that (QR) holds.

(i) If (\ref{monqrqs}) holds, then for $\tilde U=\int_0^1 A_t dt$ we have
the quasi-linear model holding
\begin{equation*}
\text{ $Y=X^\top \beta(\tilde U)$ a.s., $\tilde U \sim U(0,1)$ and ${\mathrm{%
E}}(X \mid \tilde U)={\mathrm{E}}(X)$. }
\end{equation*}
Moreover, $\tilde U$ solves the problem of correlation maximization
problem with a mean independence constraint:
\begin{equation}  \label{maxcorrmi}
\max \{ {\mathrm{E}}(VY): \; V \sim U(0,1) , \ {\mathrm{E}}(X \mid V)={%
\mathrm{E}}(X) \}.
\end{equation}

(ii) The quasi-linear representation above is unique almost surely. That is,
if we also have $Y= \overline{\beta}(\overline{U})^\top X$ with $\overline{U}
\sim U(0,1), {\mathrm{E}} (X \mid \overline{U}) = {\mathrm{E}} X$, $u
\longmapsto X^{\top }\overline{\beta}(u)$ is increasing in $u \in (0,1)$ a.s.,
then $\tilde U=\overline{U}$ and $X^{\top }\beta(\tilde U)=X^{\top }%
\overline{\beta}(\overline{U})$ a.s.

(iii) Consequently, if the conditional quantile function is linear, namely $%
Q_{Y|X}(u) = X^\top \beta_0(u)$, so that $Y= X^\top \beta_0(U)$, then the
latent factors in the quasi-linear and linear specifications coincide,
namely $U= \tilde U$, and so do the model coefficients, namely $\beta_0(U) =
\beta(U)$.
\end{theorem}


\section{Implementation of Vector Quantile Regression\label%
{sec:implementation}}

In order to implement VQR in practice, we employ discretization of the problem,
namely we approximate the distribution $F_{YX} $ of the outcome-regressor
vector $\left( Y, X\right)$ and $F_U$ of the vector rank $U$ by discrete
distributions $\nu$ and $\mu$, respectively. For example, for estimation
purposes we can approximate $F_{YX}$ by an empirical distribution of the
sample, and the distribution $F_U$ of $U$ by a finite grid.

Let $y_{i} \in \mathbb{R}^{d}$ denote values of outcomes and $x_{i}\in
\mathbb{R}^{p}$ of regressors for $1\leq i\leq n$; we assume the first
component of $x_{i}$ is $1$. For estimation purposes, we assume these values
are obtained as a random sample from distribution $F_{YX}$, and so each
observation receives a point mass $\nu_i =1/n$. When we perform computation
for theoretical purposes, we can think of these values as grid points, which
are not necessarily obtained as a random sample, and so each observation
receives a point mass $\nu_i$ which does not have to be $1/n$. We also set
up a collection of grid points $u_k$, for $k=1,...,m$, for values of the
vector rank $U$, and assign the probability mass $\mu_k$ to each of the
point. For example, if $U \sim U(0,1)^d$ and we generate values $u_k$ as a
random sample or via a uniformly spaced grid of points, then $\mu_k=1/m$.

Thus, let $\mathsf{Y}$ be the $n\times d$ matrix with row vectors $y_{j}^\top
$ and $\mathsf{X}$ the $n\times r$ matrix of row vectors $x_{j}^\top$; the
first column of this matrix is a vector of ones. Let $\nu $ be a $n\times 1$
matrix such that $\nu _{i}$ is the probability attached to a value $%
\left(x_{i},y_{i}\right) $, so that $\nu _{i}\geq 0$ and $\sum_{i=1}^n \nu
_{i}=1$. Let $m$ be the number of points in the support of $\mu$. Let $%
\mathsf{U}$ be a $m\times d$ matrix, where the $i$th row denoted by $%
u_{i}^\top$. Let $\mu $ be a $m\times 1$ matrix such that $\mu _{k}$ is the
probability weight of ${u_k}$ (hence $\mu _{i}\geq 0$ and $\sum_{k}\mu _{k}=1
$).

We are looking to find an $m\times n$ matrix $\pi $ such that $\pi _{ij}$ is
the probability mass attached to $\left( u_{i},x_{j},y_{j}\right) $ which
maximizes
\begin{equation*}
\sum_{ij}\pi _{ij}y_{j}^\top u_{i}=\text{Tr}(\mathsf{U}^{\top }\pi \mathsf{Y}%
)
\end{equation*}%
subject to constraint $\pi^\top1_{m}=\nu $, where $1_{m}$ is a $m\times 1$
vector of ones, and subject to constraints $\pi 1_{n}=\mu $ and $\pi \mathsf{%
X}=\mu \nu ^\top\mathsf{X}$.

Hence, the discretized VQR program is given in its primal form by
\begin{eqnarray}
&&\max_{\pi \geq 0}\text{Tr}\left( \mathsf{U}^{\top }\pi \mathsf{Y}\right) :
\quad \pi ^{\top }1_{m}=\nu ~\left[ \psi \right] \quad \pi \mathsf{X}=\mu
\nu ^\top\mathsf{X}~\left[ b\right],  \label{MKQRDiscrPrimal}
\end{eqnarray}%
where the square brackets show the associated Lagrange multipliers, and in
its dual form by%
\begin{eqnarray}
&&\min_{\psi ,b}\psi ^{\top }\nu +\nu ^\top\mathsf{X}b^{\top }\mu: \quad
\psi 1_{m}^{\top }+\mathsf{X}b^{\top }\geq \mathsf{Y}\mathsf{U}^{\top }~%
\left[ \pi ^{\top }\right],  \label{MKQRDiscrDual}
\end{eqnarray}%
where $\psi $ is a $n\times 1$ vector, and $b$ is a $m\times r$ matrix.

Problems (\ref{MKQRDiscrPrimal}) and (\ref{MKQRDiscrDual}) are two linear
programming problems dual to each other. However, in order to implement them
on standard numerical analysis software such as R or Matlab coupled with a
linear programming software such as Gurobi, we need to convert matrices into
vectors. This is done using the $\text{vec}$ operation, which is such that
if $A$ is a $p\times q$ matrix, $\text{vec}(A)$ is a column vector of size
$pq$ such that $\text{vec}\left( A\right) _{i+p\left( j-1\right) }=A_{ij}$. The
use of the Kronecker product is also helpful. Recall that if $A$ is a $%
p\times q$ matrix and $B$ is a $p^{\prime }\times q^{\prime }$ matrix, then
the Kronecker product $A\otimes B$ is the $pp^{\prime }\times qq^{\prime}$
matrix such that for all relevant choices of indices $i,j,k,l$, $\left(
A\otimes B\right) _{i+p\left( k-1\right) ,j+q\left( l-1\right)
}=A_{ij}B_{kl}.$ The fundamental property linking Kronecker products and the
$\text{vec}$ operator is $\text{vec}\left( BXA^{T}\right) =\left( A\otimes
B\right) \text{vec}\left( X\right).$

Introduce $\text{vec}{\pi }=\text{vec}\left( \pi \right) $, the optimization
variable of the \textquotedblleft vectorized problem\textquotedblright .
Note that the variable $\text{vec}{\pi }\ $is a $mn$-vector. Then we
rewrite the objective function,  $\mathrm{Tr}\left( \mathsf{U}^{\top }\pi
\mathsf{Y}\right) =\text{vec}{\pi }^{\top }\text{vec}\left( \mathsf{U}%
\mathsf{Y}^{\top }\right) $; as for the constraints, $\text{vec}\left(
1_{m}^{\top }\pi \right) =\left( I_{n}\otimes 1_{m}^{\top }\right) \text{vec}%
{\pi }$ is a $n\times 1$ vector; and $\text{vec}\left( \pi \mathsf{X}\right)
=\left( \mathsf{X}^{\top }\otimes I_{m}\right) \text{vec}{\pi }$ is a $
mr$-vector. Thus we can rewrite the program (\ref{MKQRDiscrPrimal})
as:
\begin{equation}
\begin{split}
& \max_{\text{vec}{\pi }\geq 0}\text{vec}\left( \mathsf{U}\mathsf{Y}^{\top
}\right) ^{\top }\text{vec}{\pi }: \\
& \left( I_{n}\otimes 1_{m}^{\top }\right) \text{vec}{\pi }=\text{vec}\left(
\nu \right) \\
& \left( \mathsf{X}^{\top }\otimes I_{m}\right) \text{vec}{\pi }=\text{vec}%
\left( \mu \nu^\top \mathsf{X}\right)
\end{split}%
\end{equation}%
which is a LP problem with $mn$ variables and $mr+n$ constraints. The
constraints $\left( I_{n}\otimes 1_{m}^{\top }\right) $ and $\left( \mathsf{X%
}^{\top }\otimes I_{m}\right) $ are very sparse, which can be taken
advantage of from a computational point of view.


\section{Empirical Illustration}

We demonstrate the use of the approach on a classical application of
Quantile Regression since \cite{KB82}: Engel's (\cite{engel}) data on
household expenditures, including 199 Belgian working-class households
surveyed by Ducpetiaux (\cite{Ducpetiaux}), and 36 observations from all
over Europe surveyed by Le Play (\cite{LePlay}). Due to the univariate
nature of classical QR, \cite{KB82} limited their focus on the regression of
food expenditure over total income. But in fact, Engel's dataset is richer
and classifies household expenses in nine broad categories: 1. Food; 2.
Clothing; 3. Housing; 4. Heating and lighting; 5. Tools; 6. Education; 7.
Safety; 8. Medical care; and 9. Services. This allows us to have a
multivariate dependent variable. While we could in principle have $d=9$, we
focus for illustrative purposes on a two-dimensional dependent variable ($d=2
$), and we choose to take $Y_{1}$ as food expenditure ( category \#1) and $%
Y_{2}$ as housing and domestic fuel expenditure (category \#2 plus category
\#4). We take $X=\left( X_{1},X_{2}\right) $ with $X_{1}=1$ and $X_{2}$%
$=$ the total expenditure (income) as an explanatory variable.



\subsection{One-dimensional VQR\label{par:App1D}}

To begin with, we run a pair of one dimensional VQRs, where we regress
 $Y_{1}$ on $X$,  and $Y_{2}$ on $X$. We plot the results in Figure~\ref{Figure:Univariate}; the curves drawn here are $u\rightarrow \beta (u)^\top x$ for five
percentiles of the income $x_{-1}$ (0\%, 25\%, 50\%, 75\%, 100\%), and the
corresponding probabilistic representations are%
\begin{equation}
Y_{1}=\beta _{1}\left( U_{1}\right) ^{\top }X\text{ and }Y_{2}=\beta
_{2}\left( U_{2}\right) ^{\top }X  \label{rep1}
\end{equation}%
with $U_{1} \mid X \sim \mathcal{U}\left( \left[ 0,1\right] \right) $ and $U_{2} \mid X \sim
\mathcal{U}\left( \left[ 0,1\right] \right) $. Here, $U_{1}$ is interpreted
as a propensity to consume food, while $U_{2}$ is interpreted as a propensity
to consume the housing good. Note that in general, $U_{1}$ and $U_{2}$ are
not independent; in other words, the distribution of $\left(
U_{1},U_{2}\right) $ differs from $\mathcal{U}(\left[ 0,1\right] ^{2})$. In
fact, the distribution of $\left( U_{1},U_{2}\right) $ is called the \emph{%
copula} associated to the conditional distribution of $\left(
Y_{1},Y_{2}\right) $ conditional on $X$.

As explained above, when $d=1$, VQR is very closely connected to classical
quantile regression. Hence, in Figure~\ref{Figure:Univariate}, we also draw
the classical quantile regression (in red). In each case, the curves exhibit
very little difference between classical quantile regression and vector
quantile regression. Small differences occur, since vector quantile
regression in the scalar case can be shown to impose the fact that map $%
t\rightarrow A_{t}$ in~(\ref{dt}) is nonincreasing, which is not necessarily
the case with classical quantile regression under misspecification in
population, or even under specification in sample. As can be seen in Figure~%
\ref{Figure:Univariate}, the difference, however, is minimal.

From the plots in Figure~\ref{Figure:Univariate}, it is also apparent that
one-dimensional VQR\ can also suffer from the \textquotedblleft crossing
problem,\textquotedblright\ namely the fact that $\beta (t)^{\top }x$ may
not be monotone with respect to $t$. Indeed, the fact that $t\rightarrow
A_{t}$ is nonincreasing fails to imply the fact that $t\rightarrow \beta
(t)^{\top }x$ is nondecreasing. There exist procedures to repair the
crossing problem, see \cite{CFG}. However, we see that the crossing problem
is modest in the current example.

Running a pair of one-dimensional Quantile Regressions is interesting, but
it does not immediately convey the information about the joint conditional
dependence in $Y_{1}$ and $Y_{2}$ (given $X$). In other words,
representations~(\ref{rep1}) are not informative about the joint propensity
to consume food and income. One could also wonder whether food and housing
are locally complements (respectively locally substitute), in the sense
that, conditional on income, an increase in the food consumption is likely
to be associated with an increase (respectively a decrease) in the
consumption of the housing good. All these questions can be immediately
answered with higher-dimensional VQR.

\begin{figure}[h]
\begin{center}
\epsfig{figure=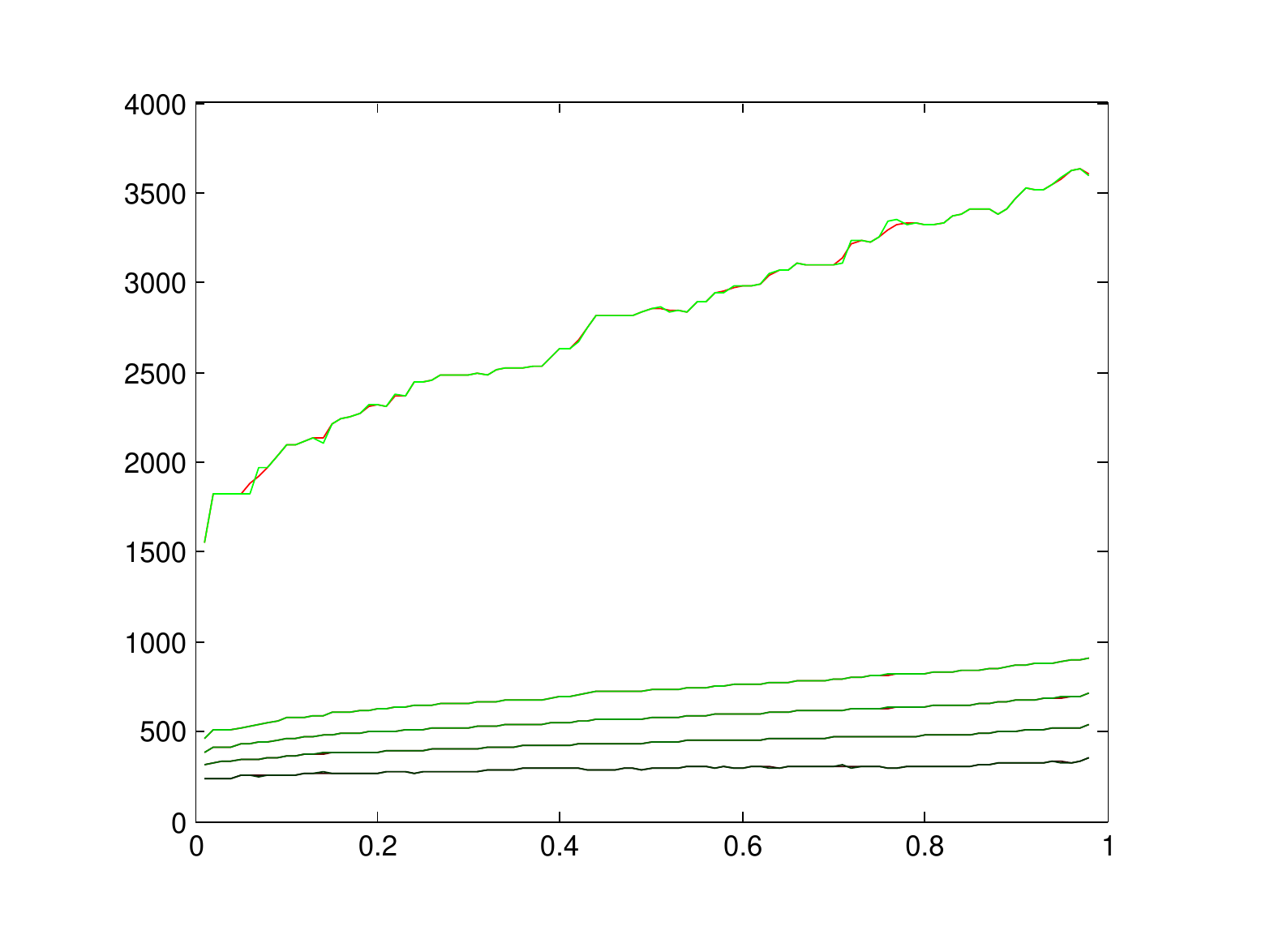,width=2.4in} \epsfig{figure=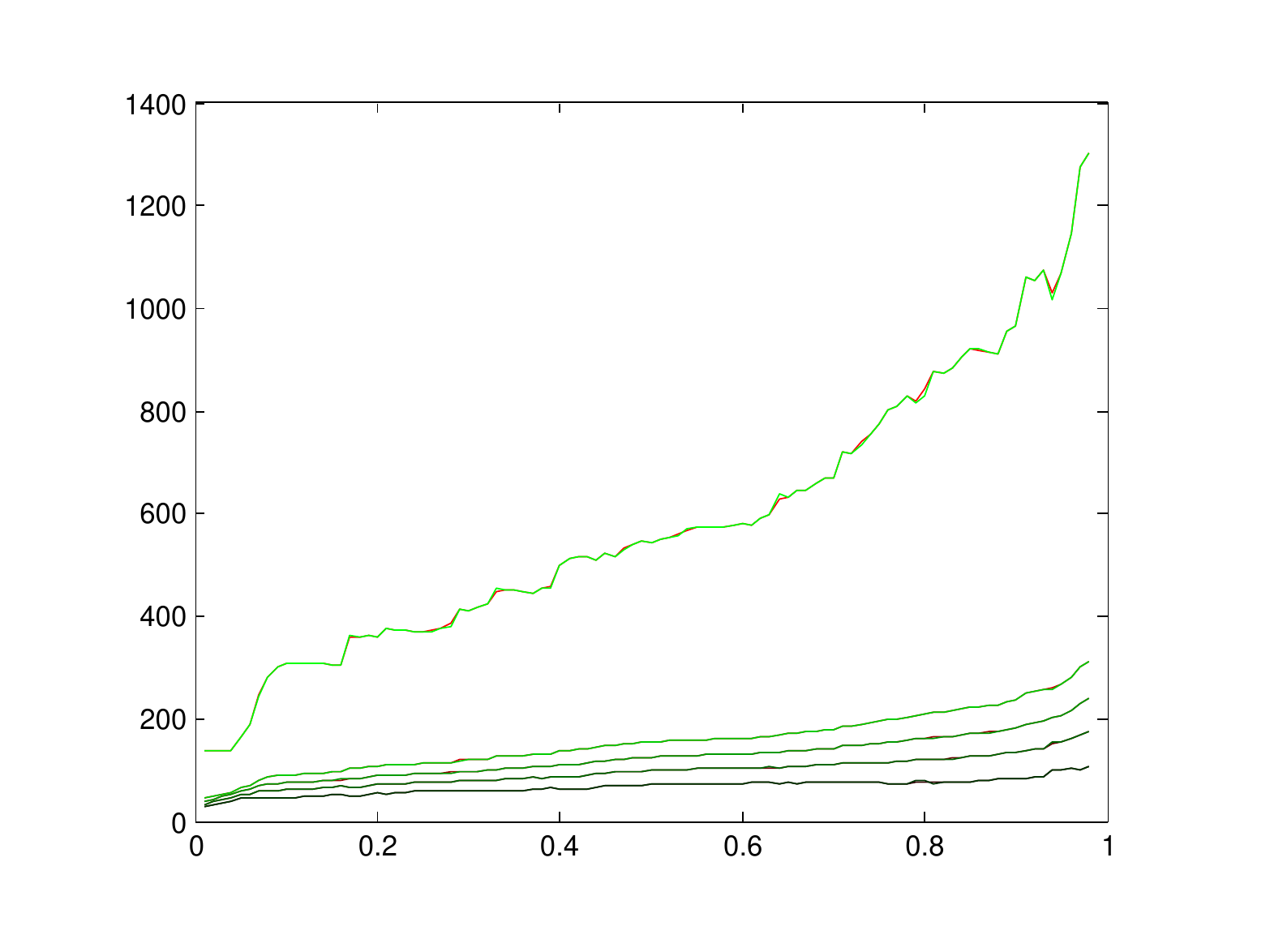,width=2.4in}
\par
\end{center}
\caption{Classical quantile regression (red) and
one-dimensional vector quantile regression (green) with income as
explanatory variable and with: (i) Food expenditure as dependent variable
(Left) and (ii) Housing expenditure as dependent variable (Right). The maps $%
t \to \protect\beta(t)^\top x$ are plotted for five values of income $x_{-1}$
(quartiles). }
\label{Figure:Univariate}
\end{figure}

\subsection{Two dimensional VQR}

In contrast, the two-dimensional vector quantile regression with $Y=\left(
Y_{1},Y_{2}\right) $ as a dependent variable yields a representation
\begin{eqnarray*}
Y_{1}= \beta_1(U_{1},U_{2} )^\top X  = \frac{\partial B}{\partial u_{1}}\left( U_{1},U_{2}\right) ^{\top }X, \\
Y_{2}=\beta_2(U_{1},U_{2} )^\top X   = \frac{\partial B}{\partial u_{2}}\left( U_{1},U_{2}\right)
^{\top }X,  \label{rep2}
\end{eqnarray*}%
where $\left( U_{1},U_{2}\right) \mid X  \sim \mathcal{U}(\left[ 0,1\right]
^{2})$. Let us make a series of remarks.

First, $U_{1}$ and $U_{2}$ have an interesting interpretation: $U_{1}$ is a
\emph{propensity for food expenditure}, while $U_{2}$ is a \emph{propensity
for domestic (housing and heating) expenditure}. Let us explain this
denomination. If VQR\ is correctly specified, then $\Phi _{x}\left( u\right)
= B\left( u\right) ^{\top }x$ is convex with respect to $u$, and $%
Y=\nabla _{u}\Phi _{X}\left( U\right) $, which implies in particular that
\begin{equation*}
\partial /\partial u_{1}\left( \partial \Phi _{x}\left( u_{1},u_{2}\right)
/\partial u_{1}\right) =\partial ^{2}\Phi _{x}\left( u_{1},u_{2}\right)
/\partial u_{1}^{2}\geq 0.
\end{equation*}%
Hence an increase in $u_{1}$ keeping $u_{2}$ constant leads to an increase
in $y_{1}$. Similarly, an increase in $u_{2}$ keeping $u_{1}$ constant leads
to an increase in $y_{2}$.

Second, the quantity $U(x,y)=Q^{-1}_{Y \mid X}(y,x)$ is a measure of joint
propensity of expenditure $Y=y$ conditional on $X=x$.
This is a way of rescaling the conditional distribution of $Y$ conditional
on $X=x$ into the uniform distribution on $\left[ 0,1\right] ^{2}$. If VQR
is correctly specified, then $\left( U_{1},U_{2}\right) $ is independent
from $X$, so that $U\left( X,Y\right) \sim F_U =\mathcal{U}(\left[ 0,1\right]
^{2})$. In this case, $\Pr \left( U\left( X,Y\right) \geq u_{1},U\left(
X,Y\right) \geq u_{2}\right) =\left( 1-u_{1}\right) \left( 1-u_{2}\right) $
can be used to detect \textquotedblleft nontypical\textquotedblright\
values of $\left( y_{1},y_{2}\right) $.

Third, representation~(\ref{rep2}) may also be used to determine if $Y_{1}$
and $Y_{2}$ are local complements or substitutes. Indeed, if VQR\ is
correctly specified and $\left( Y_{1},Y_{2}\right) $ are independent
conditional on $X$, then $B\left( u_{1},u_{2}\right) =B_{1}\left(
u_{1}\right) +B_{2}\left( u_{2}\right) $, so that the cross derivative $%
\partial ^{2}B\left( u_{1},u_{2}\right) /\partial u_{1}\partial u_{2}=0$. In
this case, (\ref{rep2}) becomes $Y_{1}=\frac{\partial B_{1}}{\partial u_{1}}%
\left( U_{1}\right) ^{\top }X$ and $Y_{2}=\frac{\partial B_{2}}{\partial
u_{2}}\left( U_{2}\right) ^{\top }X$, which is equivalent to two
single-dimensional quantile regressions. In this case, conditional on $X$,
an increase in $Y_{1}$ is not associated to an increase or a decrease in $%
Y_{2}$. On the contrary, when $\left( Y_{1},Y_{2}\right) $ are no longer
independent conditional on $X$, then the term $\partial ^{2} B\left(
u_{1},u_{2}\right) /\partial u_{1}\partial u_{2}$ is no longer zero. Assume
it is positive. In this case, an increase in the propensity to consume food $%
u_{1}$ not only increases the food consumption $y_{1}$, but also the housing
consumption $y_{2}$, which we interpret by saying that food and housing are
local complements.

\begin{figure}[h]
\begin{center}
\includegraphics[width=5in,  trim={0 6cm 0 6cm},clip]{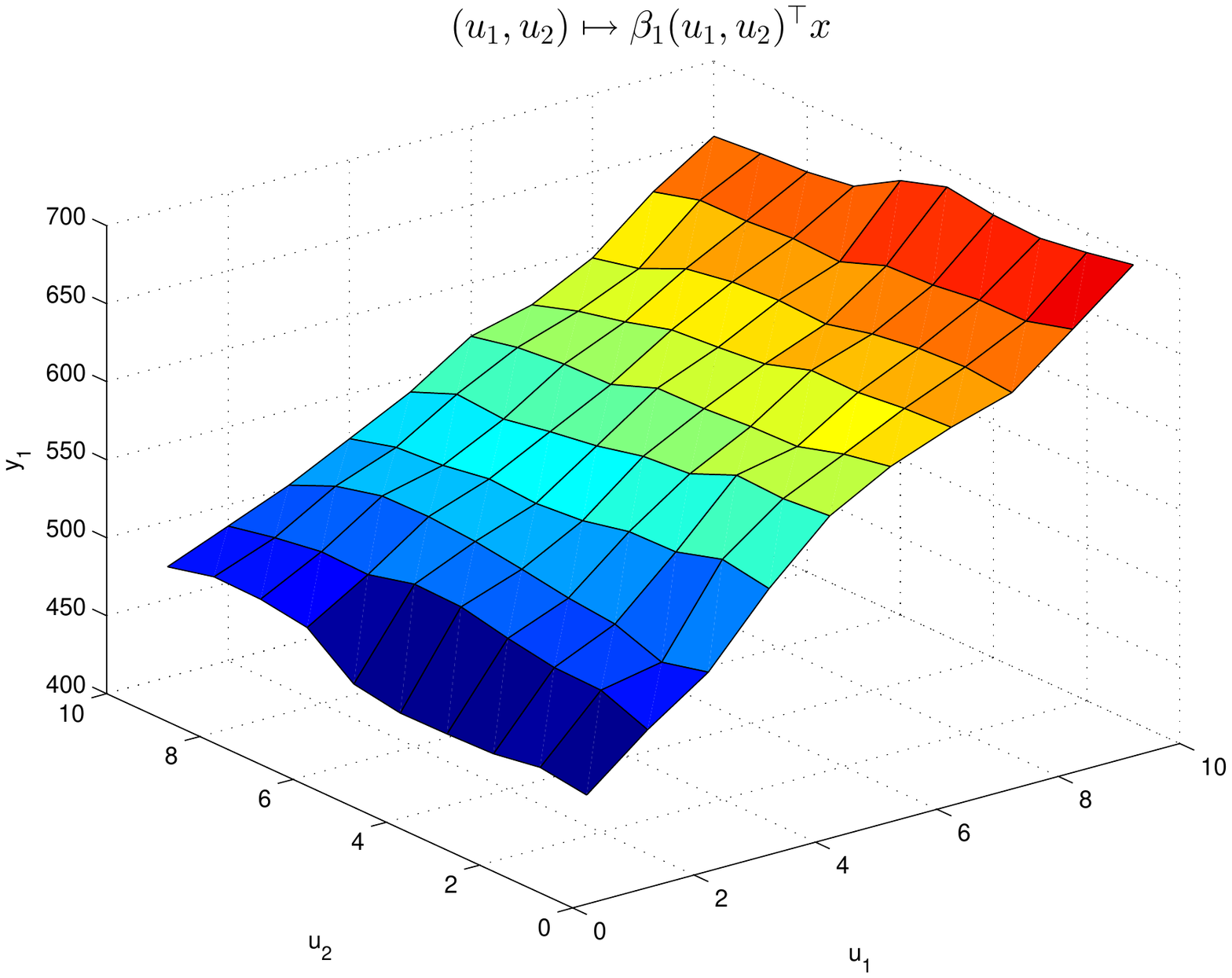} %
\includegraphics[width=5in,  trim={0 6cm 0 6cm},clip]{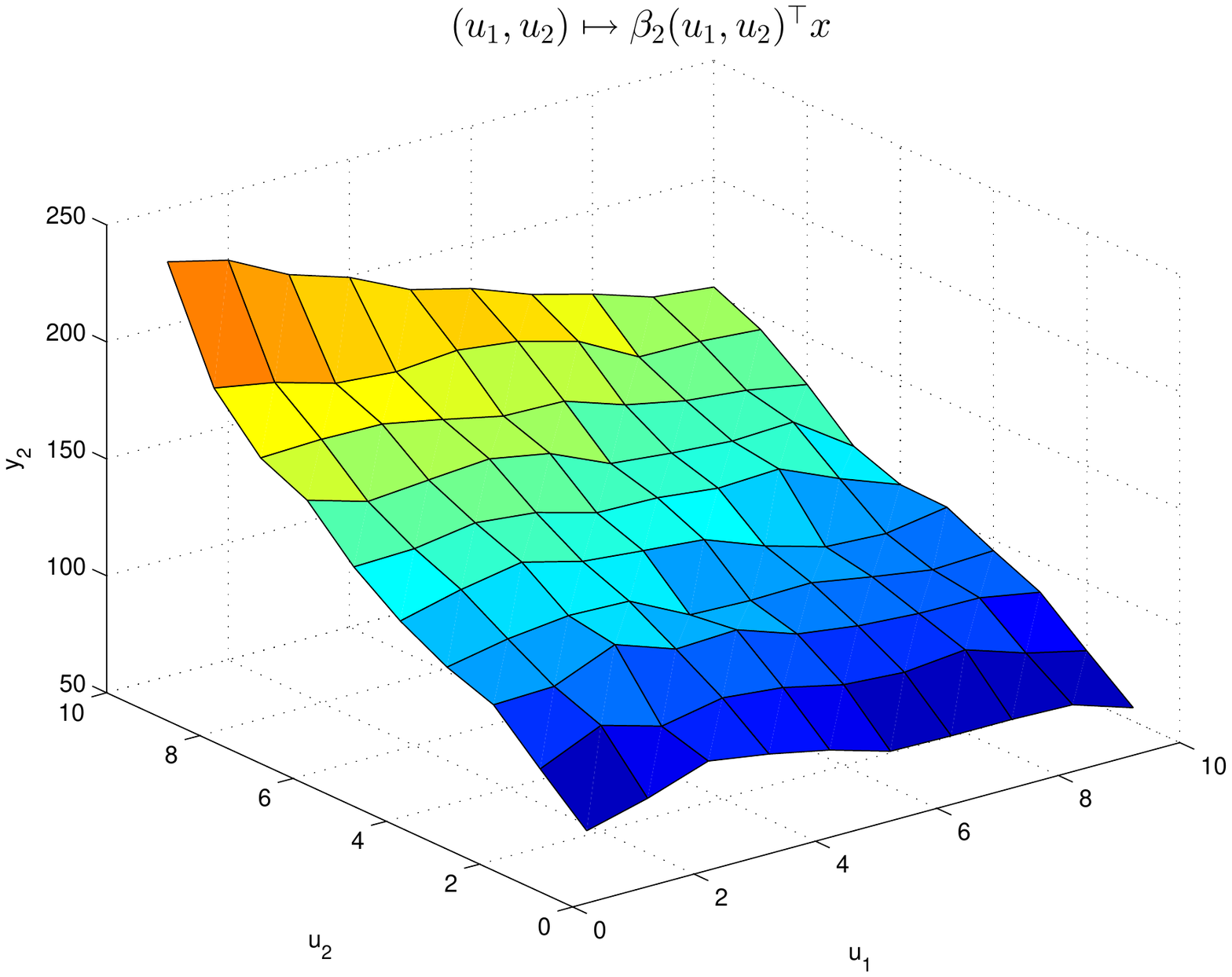}
\end{center}
\caption{Predicted outcome conditional on total expenditure equal to median
value, that is $X_{2}=883.99$. Top: food expenditure, Bottom: housing
expenditure.}
\label{Figure:Multivariate}
\end{figure}
Going back to Engel's data, in Figure~\ref{Figure:Multivariate}, we set $%
x=(1,~883.99)$, where $x_{2}=883.99$ is the median value of the total
expenditure $X_{2}$, and we are able to draw the two-dimensional
representations. 

The top pane expresses $Y_{1}$ as a function of $U_{1}$ and $U_{2}$, while
the bottom pane expresses $Y_{2}$ as a function of $U_{1}$ and $U_{2}$. The
insights of the two-dimensional representation become apparent. One sees
that while $Y_{1}$ covaries strongly with $U_{1}$ and $Y_{2}$ covaries
strongly with $U_{2}$, there is a significant and negative
cross-covariation: $Y_{1}$ covaries negatively with respect to $U_{2}$,
while $Y_{2}$ covaries negatively with $U_{1}$. The interpretation is that,
for a median level of income, the food and housing goods are local
substitutes. This makes intuitive sense, given that food and housing goods
account for a large share of the surveyed households' expenditures.

\newpage
\appendix

\section{Proofs for Section 2}

\subsection{Proof of Theorem \protect\ref{Thm: Conditional Brenier}}

The first assertion of the theorem is a consequence of the refined version
of Brenier's theorem given by \cite{McCann1985} (as, e.g, stated in \cite%
{Villani03}, Theorem 2.32), which we apply for each $z\in \mathcal{Z}$. In
particular, this implies that for each $z\in \mathcal{Z}$, the map $u\longmapsto
Q_{Y\mid Z}(u,z)$ is Borel-measurable. The  Borel measurability of $(u,z) \longmapsto(Q_{Y \mid Z}(u,z),z) $ follows from the measurability of conditional probabilities and
standard measurable selection arguments \cite{ektem}; for details, we refer to Appendix~\ref{app:measurability}.

Next we are going to show that the probability law of $(Q_{Y \mid Z}(V,Z),Z)$ is
the same as that of $(Y,Z)$. Indeed,  for any rectangle $A \times B \subset \mathbb{R}%
^{d+k} $,
\begin{eqnarray*}
&& \P ( (Y,Z) \in A\times B) = \int_B \left [ \int_{A} F_{Y\mid Z}( dy, z) %
\right] F_Z(dz) \\
&& = \int_B \left [\int 1\{ (Q_{Y \mid Z}(u,z) \in A\} F_U(du) \right ] d
F_Z(dz) \\
&& = \P ( (Q_{Y \mid Z} (V,Z),Z) \in A\times B) ,
\end{eqnarray*}
where the first equality relies on the regularity of the conditional distribution function,
and the penultimate equality follows from the previous paragraph. Since
probability measure defined over rectangles uniquely pins down the probability measure on all Borel sets via Caratheodory's extension theorem, the claim follows.

To show the second assertion we invoke Dudley-Phillip's (\cite{DP}) coupling
result given in their Lemma 2.11.

\begin{lemma}[Dudley-Phillip's coupling]
Let $S$ and $T$ be Polish spaces and $Q$ a law on $S\times T$, with marginal
law $\mu$ on $S$. Let $(\Omega, \mathcal{A}, \P )$ be a probability space
and $J$ a random variable on $\Omega$ with values in $S$ and $J \sim \mu$.
Assume there is a random variable $W$ on $\Omega$, independent of $J$, with
values in a Polish space $R$ and law $\nu$ on $R$ having no atoms. Then
there exists a random variable $I$ on $\Omega$ with values in $T$ such that $%
(J,I) \sim Q$.
\end{lemma}

First we recall that our probability space has the form:
\begin{equation*}
(\Omega, \mathcal{A},\P )=(\Omega_0, \mathcal{A}_0, \P _0) \times (\Omega_1,
\mathcal{A}_1, \P _1) \times ( (0,1), B(0,1), \text{Leb}),
\end{equation*}
where $(0,1), B(0,1), \text{Leb})$ is the canonical probability space,
consisting of the unit segment of the real line equipped with Borel sets and
the Lebesgue measure. We use this canonical space to carry $W$, which is
independent of any other random variables appearing below, and which has the
uniform distribution on $R=[0,1]$. The space $R= [0,1]$ is Polish and the
distribution of $W$ has no atoms.

Next we apply the lemma to $J = (Y,Z)$ to show existence of $I = U$, where
both $J$ and $I$ live on the probability space $(\Omega, \mathcal{A}, \P )$
and that obeys the second assertion of the theorem. The variable $J$ takes
values in the Polish space $S= \mathbb{R}^d \times \mathbb{R}^k$, and the
variable $I$ takes values in the Polish space $T= \mathbb{R}^d$.

Next we describe a law $Q$ on $S\times T$ by defining a triple $(Y^*, Z^*,
U^*)$ that lives on a suitable probability space. We consider a random
vector $Z^*$ with distribution $F_Z$, a random vector $U^* \sim F_U$,
independently distributed of $Z^*$, and $Y^* = Q_{Y\mid Z}(U^*, Z^*)$
uniquely determined by the pair $(U^*,Z^*)$, which completely characterizes
the law $Q$ of $(Y^*,Z^*, U^*)$. In particular, the triple obeys $Z^* \sim
F_Z$, $U^*|Z^* \sim F_U$ and $Y^*\mid Z^*=z \sim F_{Y \mid Z}(\cdot,z)$.
Moreover, the set $\{(y^*,z^*, u^*): \|y^* - Q_{Y\mid Z}(u,z^*)\|=
0\}\subset S\times T$ is assigned probability mass 1 under $Q$.

By the lemma quoted above, given $J$, there exists an $I=U$, such that $%
(J,I)\sim Q$, but this implies that $U|Z \sim F_U$ and that $\|Y -Q_{Y\mid
Z}(U,Z)\|= 0$ with probability 1 under $\P $. \qed

\subsection{Proof of Theorem \protect\ref{Thm: Inverse}}

We condition on $Z=z$. By reversing the roles of $V$ and $Y$, we can apply
Theorem \ref{Thm: Conditional Brenier} to claim that there exists a map $y
\longmapsto Q^{-1}_{Y\mid Z}(y,z)$ with the properties stated in the theorem
such that $Q^{-1}_{Y\mid Z}(Y,z)$ has distribution function $F_U$,
conditional on $Z=z$. Hence for any test function $\xi: {\mathbb{R}}^d \to {%
\mathbb{R}}$ such that $\xi \in C_b({\mathbb{R}}^d)$ we have
\begin{equation*}
\int \xi ( Q^{-1}_{Y\mid Z}( Q_{Y \mid Z} (u,z ),z ) )F_U(d u) = \int \xi(u)
F_U(du).
\end{equation*}
This implies that for $F_U$-almost every $u$, we have $Q^{-1}_{Y\mid Z}(
Q_{Y \mid Z} (u,z ),z) = u.$ Hence $\P $-almost surely,
$Q^{-1}_{Y\mid Z}(Y,Z) = Q^{-1}_{Y\mid Z}( Q_{Y \mid Z} (U,Z ) ,Z) = U.
$ Thus we can set $U = Q^{-1}_{Y\mid Z}(Y,Z)$ $\P $-almost surely in Theorem %
\ref{Thm: Conditional Brenier}. \qed

\subsection{Proof of Theorem \protect\ref{thm: variational characterization}}

The result follows from \cite{Villani03}, Theorem 2.12. \qed

\section{Proofs for Section 3}

\subsection{Proof of Theorem \protect\ref{theorem: linear}}

We first establish part(i). We have a.s.
\begin{equation*}
Y=\nabla \Phi_X(\widetilde U), \mbox{ with } \Phi_X(u)=B(u)^\top X.
\end{equation*}
For any $V \sim F_U$ such that ${\mathrm{E}}(X\vert V)=\mathrm{E}(X)$, and $%
\Phi_x^*(y) := \sup_{v \in \mathcal{U}} \{ v^\top y - \Phi_x(v) \}$, we have by the mean independence
\begin{equation*}
{\mathrm{E}} [\Phi_X(V)+\Phi_X^*(Y)] = {\mathrm{E}} B(V)^\top \mathrm{E}(X) + {%
\mathrm{E}} \Phi_X^*(Y) := M,
\end{equation*}
where $M$ depends only on $F_U$. We have by Young's inequality
\begin{equation*}
V^\top Y \le \Phi_X(V)+\Phi_X^*(Y),
\end{equation*}
but $Y=\nabla \Phi_X(\widetilde U)$ a.s. implies that a.s.
\begin{equation*}
\widetilde U^\top Y = \Phi_X(\widetilde U)+\Phi_X^*(Y),
\end{equation*}
so taking expectations gives
\begin{equation*}
{\mathrm{E}} V^\top Y \leq M = {\mathrm{E}} \widetilde U^\top Y, \quad
\forall V \sim F_U : {\mathrm{E}}(X\vert V)=\mathrm{E}(X),
\end{equation*}
which yields the desired conclusion.

We next establish part(ii). We can argue similarly as above to show that
$Y= \bar \beta(\overline{U})^\top X = \nabla \bar \Phi_X(\overline{U}),$
for some convex potential $u \longmapsto \bar \Phi_x(u)=\bar B(u)^\top x$,
and that for $\bar \Phi_x^*(y) := \sup_{v \in \mathcal{U}} \{ v^\top y -
\Phi_x(v) \}$ we have a.s.
\begin{equation*}
\overline{U}^\top Y = \bar \Phi_X(\overline{U})+ \bar \Phi_X^*(Y).
\end{equation*}

Using the fact that $\tilde U \sim \overline{U}$ and the fact that
mean-independence gives ${\mathrm{E}}(B(\tilde U)^\top X)={\mathrm{E}}B(\tilde U)^\top
\mathrm{E}(X)={\mathrm{E}} B(\bar U) ^\top {\mathrm{E}}(X) = {\mathrm{E}}(B(%
\overline{U})^\top X)$, we have
\begin{equation*}
{\mathrm{E}}(\widetilde U^\top Y)={\mathrm{E}}(  \Phi^*_X(Y)+B(\widetilde U)^\top X)=
{\mathrm{E}}(\Phi^*_X(Y)+B(\overline{U})^\top X) \ge {\mathrm{E}}(\overline{U}^\top
Y),
\end{equation*}
where we used Young's inequality again.  Reversing the role of $\widetilde U$ and $\overline{U}$ and using the same reasoning, we also conclude that ${\mathrm{E}}%
(\widetilde U^\top Y)\le {\mathrm{E}}(\overline{U}^\top Y)$, and hence then
\begin{equation*}
{\mathrm{E}}(\overline{U}^\top Y)= {\mathrm{E}}( \Phi^*_X(Y)+B(\overline{U})^\top X)
\end{equation*}
so that, thanks to the Young's inequality again,
\begin{equation*}
\Phi^*_X(y)+B(u)^\top x \ge u^\top y, \; \forall (u,x,y)\in \mathcal{U}%
\mathcal{X}\mathcal{Y},
\end{equation*}
we have
\begin{equation*}
\Phi^*_X(Y)+B(\overline{U})^\top X =\overline{U}^\top Y, \; \mbox{ a.s. },
\end{equation*}
which means that $\overline{U} = \arg \max_{u \in \mathcal{U}} \{u^\top
Y-B(u)^\top X\}$ a.s., which by strict concavity admits $\tilde U$ as the unique
solution a.s. This proves that $\tilde U=\overline{U}$ a.s. and thus we have $(%
\overline{\beta}(\tilde U)-\beta(\tilde U))^\top X =0$ a.s.

The part (iii) is a consequence of part (i). Note that by part (ii) we have
that $\tilde U = U$ a.s. and $(\beta(U)-\beta_0(U))^\top X =0$ a.s. Since $U$
and $X$ are independent, we have that, for $e_1,...,e_p$ denoting vectors of
the canonical basis in ${\mathbb{R}}^p$, and each $j \in \{1,\ldots, p\}$:
\begin{eqnarray*}
&& 0 = {\mathrm{E}} \left (e_j^\top (\beta(U)-\beta_0(U))^\top X X^\top
(\beta(U)-\beta_0(U) ) e_j \right ) \\
&& = {\mathrm{E}} \left ( e_j^\top (\beta(U)-\beta_0(U))^\top {\mathrm{E}} X
X^\top (\beta(U)-\beta_0(U)) e_j \right ) \\
&& \geq \mathrm{mineg} ( {\mathrm{E}} X X^\top) {\mathrm{E}} \left (
\|(\beta(U)-\beta_0(U)) e_j\|^2 \right ).
\end{eqnarray*}
Since ${\mathrm{E}} X X^\top$ has full rank this implies that ${\mathrm{E}}
\|(\beta(U)-\beta_0(U) )e_j\|^2 =0$ for each $j \in \{ 1, \ldots, p\}$, which implies the rest of
the claim. \qed

\subsection{Proof of Theorem \protect\ref{theorem: dual1}}

We have that any feasible pair $(\psi, b)$ obeys the constraint
\begin{equation*}
\psi(x,y) + b(u)^\top x \geq u^\top y, \ \ \forall (y,x,u) \in \mathcal{Y}%
\mathcal{X}\mathcal{U}.
\end{equation*}
Let $\tilde U\sim F_U: {\mathrm{E}}(X \mid U) = {\mathrm{E}}(X)$ be the
solution to the primal program. Then for any feasible pair $(\psi, b)$ we
have:
\begin{equation*}
{\mathrm{E}} \psi(X,Y) + {\mathrm{E}} b(\tilde U)^\top {\mathrm{E}} X = {%
\mathrm{E}} \psi(X,Y) + {\mathrm{E}} b(\tilde U)^\top X \geq {\mathrm{E}} Y
^\top \tilde U .
\end{equation*}

 The last inequality holds as equality if
\begin{equation}  \label{eq: cool pair}
\psi(x,y)=\sup_{u\in \mathcal{U}} \{ u^\top y- b(u)^\top x\}, \ \ b(u) =
B(u),
\end{equation}
since this is a feasible pair by (QL) and since
\begin{equation*}
\psi(X,Y) + B(\tilde U)^\top X = Y ^\top \tilde U,
\end{equation*}
as shown the proof of
the previous theorem. It follows that ${\mathrm{E}} Y ^\top \tilde U$ is the optimal value and it
is attained by the pair (\ref{eq: cool pair}). \qed

\subsection{Proof of Theorem \protect\ref{qrqsdec}}

Obviously $A_t=1\Rightarrow \tilde U \ge t, \mbox{ and } \; \tilde U>t
\Rightarrow A_t=1$. Hence ${\mathrm{P}}(\tilde U\ge t)\ge {\mathrm{P}}%
(A_t=1)={\mathrm{P}}(Y> \beta(t)^\top X)=(1-t)$ and ${\mathrm{P}}(\tilde U>
t)\le {\mathrm{P}}(A_t=1)=(1-t)$ which proves that $\tilde U$ is uniformly
distributed and $1\{\tilde U>t\}$ coincides with $1\{\tilde U_t=1\}$ a.s. We
thus have ${\mathrm{E}}(X 1\{\tilde U>t\})={\mathrm{E}}(X A_t)={\mathrm{E}}
X (1-t) = {\mathrm{E}} X {\mathrm{E}} A_t$, with standard approximation
argument we deduce that ${\mathrm{E}}(Xf(\tilde U))={\mathrm{E}} X {\mathrm{E%
}} f(\tilde U)$ for every $f\in C_b([0,1], {\mathbb{R}})$, which means that $%
\mathrm{E}(X \mid \tilde U) = \mathrm{E}(X)$.

As already observed $\tilde U>t$ implies that $Y>\beta(t)^\top X$ in
particular $Y\ge \beta(\tilde U- \delta) ^\top X$ for $\delta>0$, letting $%
\delta\to 0^+$ and using the a.e. continuity of $u \longmapsto \beta(u)$ we get $%
Y\ge \beta(\tilde U) ^\top X$. The converse inequality is obtained similarly
by remaking that $\tilde U<t$ implies that $Y\le \beta(t)^\top X$.

Let us now prove that $\tilde U$ solves (\ref{maxcorrmi}). Take $V$
uniformly distributed and mean-independent from $X$ and set $V_t:=1{\{V>t \}}
$, we then have ${\mathrm{E}}(X V_t)=0$, ${\mathrm{E}}(V_t)=(1-t)$ but since
$A_t$ solves (\ref{dt}) we have ${\mathrm{E}}(V_t Y)\le {\mathrm{E}}(A_tY)$.
Observing that $V=\int_0^1 V_t dt$ and integrating the previous inequality
with respect to $t$ gives ${\mathrm{E}}(VY)\le {\mathrm{E}}(\widetilde UY)$ so that $%
\tilde U$ solves (\ref{maxcorrmi}).

Next we show part(ii). Let us define for every $t\in [0,1]$ $B(t):=\int_0^t
\beta(s)ds.$ Let us also define for $(x,y)$ in ${\mathbb{R}}^{N+1}$:
\begin{equation*}
\psi(x,y):=\max_{t\in [0,1]} \{ty-B(t)^\top x\};
\end{equation*}
thanks to monotonicity condition, the maximization program in the display above is strictly
concave in $t$ for every $y$ and each $x \in X$. We then note that
\begin{equation*}
Y=\beta(\tilde U)^\top X=\nabla B(\tilde U)^\top X \text{ a.s. }
\end{equation*}
exactly is the first-order condition for the above maximization problem when
$(x,y)=(X,Y)$. In other words, we have
\begin{equation}  \label{ineqq}
\psi(x,y)+B(t)^\top x \ge ty, \; \forall (t,x,y)\in [0,1]\times \mathcal{X}
\times {\mathbb{R}}
\end{equation}
with an equality holding a.s. for $(x,y,t)=(X,Y,\tilde U)$, i.e.
\begin{equation}  \label{eqas}
\psi(X,Y)+B(\tilde U)^\top X =\widetilde UY, \; \mbox{ a.s. }
\end{equation}

Using the fact that $\tilde U \sim \overline{U}$ and the fact that the mean
independence gives ${\mathrm{E}}(B(\tilde U)^\top X)={\mathrm{E}}(b(%
\overline{U})^\top X)=\mathrm{E}(X)$, we have
\begin{equation*}
{\mathrm{E}}(UY)={\mathrm{E}}( \psi(X,Y)+B(\tilde U)^\top X )= {\mathrm{E}}(
\psi(X,Y)+B(\overline{U})^\top X ) \ge {\mathrm{E}}(\overline{U} Y)
\end{equation*}
but reversing the role of $\tilde U$ and $\overline{U}$, we also have ${%
\mathrm{E}}(UY)\le {\mathrm{E}}(\overline{U} Y)$ and then
\begin{equation*}
{\mathrm{E}}(\overline{U} Y)= {\mathrm{E}}( \psi(X,Y)+B(\overline{U})^\top X
)
\end{equation*}
so that, thanks to inequality (\ref{ineqq})
\begin{equation*}
\psi(X,Y)+B(\overline{U})^\top X =\overline{U} Y, \; \mbox{ a.s. }
\end{equation*}
which means that $\overline{U}$ solves $\max_{t\in [0,1]}
\{tY-\varphi(t)-B(t)^\top X\}$ which, by strict concavity admits $\tilde U$
as unique solution.

Part (iii) is a consequence of Part (ii) and independence of $\tilde U$ and $%
X$. Note that by part (ii) we have that $\tilde U = U$ a.s. and that $%
(\beta( U)-\beta_0(U))^\top X =0$ a.s. Since $U$ and $X$ are independent, we
have that
\begin{eqnarray*}
&& 0 = {\mathrm{E}} \left ( (\beta(\tilde U)-\beta_0(U))^\top X X^\top
(\beta(U)-\beta_0(U)) \right ) \\
&& = {\mathrm{E}} \left ( (\beta(U)-\beta_0(U))^\top {\mathrm{E}} X X^\top
(\beta(U)-\beta_0(U) )\right ) \\
&& \geq \mathrm{mineg} ( {\mathrm{E}} X X^\top) {\mathrm{E}} \left (
\|(\beta(U)-\beta_0(U))\|^2 \right) .
\end{eqnarray*}
Since ${\mathrm{E}} X X^\top$ has full rank this implies that ${\mathrm{E}}
\|(\beta(U)-\beta_0(U) )\|^2 =0$, which implies the rest of the claim. \qed


\section{Additional Results for Section 2 and Section 3}

This section  provides a rigorous Proof of
Duality for Conditional Vector Quantiles and Linear Vector Quantile
Regression.

We claimed in the main text we stated two dual problems for CVQF and VQR
without the proof. In this sections we prove that these assertions were
rigorous. Here for simplicity of notation we assume that $X=Z$, which
entails no loss of generality under our assumption that conditioning on $Z$
and $X$ is equivalent.

We shall write $X = (1, \tilde X^\top)^\top$, where $\tilde X$ denotes the
non-constant component of vector $X$. Let $f_W$ denote the joint density of $%
(\tilde X, Y)$ with support $W=\mathcal{W}$, and $F_U$ the distribution of $%
U $ with support set $\mathcal{U}$. Assume without loss of generality that
\begin{equation*}
{\mathrm{E}} ( \tilde X) = 0.
\end{equation*}

\subsection{Duality for Conditional Vector Quantiles}

Let $k$ and $d$ be two integers, $K_1$ be a compact subset of ${\mathbb{R}}%
^k $ and $K_2$ and $K_3$ be compact subsets of ${\mathbb{R}}^d$. Let $%
\nu=F_W $ have support $\mathcal{W} \subseteq K_1\times K_2$, and we may
decompose $\nu= m\otimes \nu^{\tilde{x}}$ where $m$ denotes the first
marginal of $\nu$. We also assume that $m$ is centered i.e.
\begin{equation*}
\overline{x}:=\int_{K_1} {\tilde{x}} m(d{\tilde{x}})=0
\end{equation*}
Finally, let $\mu = F_U$ have support $\mathcal{U} \subseteq K_3$. We are
interested here in rigorous derivation for dual formulations for covariance
maximization under an independence and then a mean-independence constraint.

\subsubsection*{Duality for the independence constraint}

First consider the case of an independence constraint:
\begin{equation}  \label{mki}
\sup_{\theta\in I(\mu, \nu)} \int_{K_1\times K_2\times K_3} u^\top y \;
\theta(d{\tilde{x}}, dy, du)
\end{equation}
where $I(\mu, \nu)$ consists of the probability measures $\theta$ on $%
K_1\times K_2\times K_3$ such that ${\Pi_{X,Y}}_\# \theta=\nu$ and ${%
\Pi_{X,U}}_\# \theta=m\otimes \mu$, namely that
\begin{equation*}
\int_{K_1\times K_2\times K_3} \psi({\tilde{x}},y) \theta(d{\tilde{x}}, dy,
du)= \int_{K_1\times K_2} \psi({\tilde{x}},y) \nu(d{\tilde{x}}, dy), \;
\forall \psi\in C(K_1\times K_2),
\end{equation*}
and
\begin{equation*}
\int_{K_1\times K_2\times K_3} \varphi({\tilde{x}},u) \theta(d{\tilde{x}},
dy, du)= \int_{K_1\times K_3} \varphi({\tilde{x}},u) m(d{\tilde{x}}) \mu(
du), \; \forall \varphi \in C(K_1\times K_3).
\end{equation*}
As already noticed, given a random $(X,Y)$ such that $\mathop{\mathrm{Law}}%
\nolimits(X,Y)=\nu$, (\ref{mki}) is related to the problem of finding $U$
independent of $X$ and having law $\mu$ which is maximally correlated to $Y$%
. It is clear that $I(\mu, \nu)$ is a nonempty (take $m\otimes \nu^{\tilde{x}%
}\otimes \mu$) convex and weakly $*$ compact set so that (\ref{mki}) admits
solutions. Let us consider now:
\begin{equation}  \label{dualmki}
\inf_{(\psi,\varphi)\in C(K_1\times K_2)\times C(K_1\times K_3)}
\int_{K_1\times K_2} \psi({\tilde{x}},y) \nu(d{\tilde{x}}, dy)+
\int_{K_1\times K_3} \varphi({\tilde{x}},u) m(d{\tilde{x}}) \mu(du)
\end{equation}
subject to the constraint
\begin{equation}  \label{constrdualmki}
\psi({\tilde{x}},y)+\varphi({\tilde{x}},u)\ge u ^\top y, \; \forall ({\tilde{%
x}},y, u)\in K_1\times K_2\times K_3.
\end{equation}
Then we have

\begin{theorem}
\label{dual78} The infimum in (\ref{dualmki})-(\ref{constrdualmki})
coincides with the maximum in (\ref{mki}). This common value also coincides
with the infimum of $\int_{K_1\times K_2} \psi \nu+ \int_{K_1\times K_2}
\varphi m\otimes \mu$ taken over $L^1(\nu)\times L^1(m\otimes \mu)$
functions that satisfy the constraint (\ref{constrdualmki}).
\end{theorem}

\begin{proof}
Let us rewrite (\ref{dualmki}) in standard convex programming form as :
\begin{equation*}
\inf_{(\psi, \varphi)} F(\psi, \varphi)+G(\Lambda(\psi, \varphi))
\end{equation*}
where $F(\psi,\varphi):= \int_{K_1\times K_2} \psi \nu+ \int_{K_1\times K_3}
\varphi m\otimes \mu$, $\Lambda$ is the linear continuous map from $%
C(K_1\times K_2)\times C(K_1\times K_3)$ to $C(K_1\times K_2 \times K_3)$
defined by
\begin{equation*}
\Lambda(\psi, \varphi) ({\tilde{x}},y,u):=\psi({\tilde{x}},y)+\varphi({%
\tilde{x}},u), \; \forall ({\tilde{x}},y, u)\in K_1\times K_2\times K_3
\end{equation*}
and $G$ is defined for $\eta\in C(K_1\times K_2\times K_3)$ by:
\begin{equation*}
G(\eta) =\left\{%
\begin{array}{lll}
0 & \mbox{ if } \eta({\tilde{x}},y,u)\ge u ^\top y &  \\
+\infty & \mbox{ otherwise.} &
\end{array}
\right.
\end{equation*}
It is easy to check that the Fenchel-Rockafellar theorem (see Ekeland and
Temam \cite{ektem}) applies here so that the infimum in (\ref{dualmki})
coincides with
\begin{equation}  \label{dual89}
\sup_{\theta\in {\mathcal{M}}(K_1\times K_2\times K_3)} -F^*(\Lambda^*
\theta)-G^*(-\theta).
\end{equation}
Direct computations give that
\begin{equation*}
-G^*(-\theta)=\left\{%
\begin{array}{lll}
\int_{K_1\times K_2\times K_3} u^\top y \; \theta (d{\tilde{x}}, dy, du) & %
\mbox{ if } \theta \ge 0 &  \\
-\infty & \mbox{ otherwise.} &
\end{array}
\right.
\end{equation*}
that $\Lambda^* \theta=({\Pi_{X,Y}}_\# \theta, {\Pi_{X,U}}_\# \theta)$ and
\begin{equation*}
F^*(\Lambda^* \theta)=\left\{%
\begin{array}{lll}
0 & \mbox{ if } ({\Pi_{X,Y}}_\# \theta, {\Pi_{X,U}}_\# \theta)=(\nu,
m\otimes \mu) &  \\
+\infty & \mbox{ otherwise.} &
\end{array}
\right.
\end{equation*}
This shows that the maximization problem (\ref{dual89}) is the same as (\ref%
{mki}). Therefore the infimum in (\ref{dualmki}) coincides with the maximum
in (\ref{mki}). When one relaxes (\ref{dualmki}) to $L^1$ functions, we
obtain a problem whose value is less than that of (\ref{dualmki}) (because
minimization is performed over a larger set) and larger than the supremum in
(\ref{mki}) (direct integration of the inequality constraint), the common
value of (\ref{mki}) and (\ref{dualmki})-(\ref{constrdualmki}) therefore
also coincides with the value of the $L^1$ relaxation of (\ref{dualmki})-(%
\ref{constrdualmki}).
\end{proof}

\subsection{Duality for Linear Vector Quantile Regression}

Here we use the notation where we partition
\begin{equation*}
b(u) = (\varphi(u), v(u)^\top)^\top,
\end{equation*}
with $u \mapsto \varphi(u)$ mapping $\mathcal{U}$ to $\mathbb{R}$,
corresponding to the coefficient in front of the constant. Let $\mathcal{W}%
_o $ denote the interior of $\mathcal{W}$.

Let us denote by $(0, \overline{y})$ the mean of $f_W$:
\begin{equation*}
\int_{{\mathcal{W}_o}} \tilde{x} \; F_W (d\tilde{x}, dy)=0, \; \int_{{%
\mathcal{W}_o}} y \; F_W(d\tilde{x}, dy)=:\overline{y}.
\end{equation*}

Let us now consider the mean-independent correlation maximization problem:
\begin{equation}  \label{mkmi}
\sup_{\theta\in MI(\mu, \nu)} \int_{K_1\times K_2\times K_3} u^\top y \;
\theta(d{\tilde{x}}, dy, du)
\end{equation}
where $MI(\mu, \nu)$ consists of the probability measures $\theta$ on $%
K_1\times K_2\times K_3$ such that ${\Pi_{X,Y}}_\# \theta=\nu$, ${\Pi_U}_\#
\theta=\mu$ and according to $\theta$, ${\tilde{x}}$ is mean independent of $%
u$ i.e.
\begin{equation}  \label{mift}
\ \langle \sigma_\theta, v \rangle := \int_{K_1\times K_2\times K_3}
(v(u)^\top {\tilde{x}})\ \; \theta(d{\tilde{x}}, dy, du)=0, \; \forall
\beta\in C(K_3, {\mathbb{R}}^d).
\end{equation}
Again the constraints are linear so that $MI(\mu, \nu)$ is a nonempty convex
and weak $*$ compact set so that the infimum in (\ref{mki}) is attained. In
probabilistic terms, given $(\tilde X,Y)$ distributed according to $\nu$,
the problem above consists in finding $U$ with law $\mu$, mean-independent
of $\tilde X$ (i.e. such that ${\mathrm{E}}(\tilde X\vert U)={\mathrm{E}}%
(\tilde X)=0$) and maximally correlated to $Y$.

We claim that (\ref{mkmi}) is dual to

\begin{equation}  \label{dualmkmi}
\inf_{(\psi, \varphi, b)\in C(K_1\times K_2)\times C(K_3)\times C(K_3, {%
\mathbb{R}}^d)} \int_{K_1\times K_2} \psi({\tilde{x}},y) \nu(d{\tilde{x}},
dy)+ \int_{K_3} \varphi(u) \mu(du)
\end{equation}
subject to
\begin{equation}  \label{constrdualmkmi}
\psi({\tilde{x}},y)+\varphi(u)+v(u)^\top {\tilde{x}} \ge u ^\top y, \;
\forall ({\tilde{x}},y, u)\in K_1\times K_2\times K_3.
\end{equation}
Then we have the following duality result

\begin{theorem}
The infimum in (\ref{dualmkmi})-(\ref{constrdualmkmi}) coincides with the
maximum in (\ref{mkmi}). This common value also coincides with the infimum
of $\int_{K_1\times K_2} \psi \nu+ \int_{K_3} \varphi m\otimes \nu$ taken
over $L^1(\nu)\times L^1(\mu)\times L^1(\mu, {\mathbb{R}}^d)$ functions that
satisfy the constraint (\ref{constrdualmkmi}).
\end{theorem}

\begin{proof}
Write (\ref{dualmkmi})-(\ref{constrdualmkmi}) as
\begin{equation*}
\inf_{(\psi, \varphi, b)} F(\psi, \varphi, b)+G(\Lambda(\psi, \varphi, b))
\end{equation*}
where $F(\psi,\varphi, b):= \int_{K_1\times K_2} \psi \nu+ \int_{K_3}
\varphi \mu$, $\Lambda$ is the linear continuous map from $C(K_1\times
K_2)\times C(K_3)\times C(K_3, {\mathbb{R}}^d)$ to $C(K_1\times K_2 \times
K_3)$ defined by
\begin{equation*}
\Lambda(\psi, \varphi, b) ({\tilde{x}},y,u):=\psi({\tilde{x}}%
,y)+\varphi(u)+v(u)^\top {\tilde{x}}, \; \forall ({\tilde{x}},y, u)\in
K_1\times K_2\times K_3
\end{equation*}
and $G$ is as in the proof of theorem \ref{dual78}. For $\theta\in {\mathcal{%
M}}(K_1\times K_2\times K_3)$, one directly checks that $\Lambda^* \theta=( {%
\Pi_{\tilde X,Y}}_\# \theta, {\Pi_{U}}_\# \theta, \sigma_\theta)$ (where the
vector-valued measure $\sigma_\theta$ is defined as in (\ref{mift})). We
then get
\begin{equation*}
F^*(\Lambda^* \theta)=\left\{%
\begin{array}{lll}
0 & \mbox{ if } ({\Pi_{\tilde X,Y}}_\# \theta, {\Pi_{U}}_\# \theta,
\sigma_\theta)=(\nu, \mu, 0 ) &  \\
+\infty & \mbox{ otherwise.} &
\end{array}
\right.
\end{equation*}
We then argue exactly as in the proof of theorem \ref{dual78} to conclude.
\end{proof}

\section{Proof of the measurability claim in Theorem 2.1} \label{app:measurability} Let us denote by $\mu$ the probability law (absolutely continuous with respect to the Lebesgue measure on $\R^d$ and having a convex support with nonempty interior) of $U$ and by $\nu^z$ the conditional probability law of $Y$ given $Z=z$.

Let us also assume first that  both $U$ and $Y$ have finite second moments. Denoting by $\PPP_2(\R^k)$ the Wasserstein space of probability measures  on $\R^k$ having finite second moments endowed with the $2$-Wasserstein distance $W_2$, we recall that it is a separable metric space (see \cite{Villani03}). Now consider the set-valued map
\[\Gamma: \; \nu \in \PPP_2(\R^d) \mapsto \{\gamma\in \Pi(\mu, \nu) \ : \; \int_{\R^d\times \R^d} \vert x-y\vert^2 \mbox{d} \gamma(x,y)=W_2^2(\mu, \nu)\}\]
where $\Pi(\mu, \nu)$ is the set of $\gamma\in \PPP_2(\R^d\times \R^d)$ having $\mu$ and $\nu$ as marginals. It then follows from standard measurable selection arguments (see \cite{casval} or   chapter VIII of \cite{ektem}) that $\Gamma$ admits a Borel (with respect to the Wasserstein metric) selection which we denote  $\gamma(.)$. Note in particular that $z\mapsto \gamma(\nu^z)$ is measurable and the optimality of $\gamma(\nu^z)$ is characterized by the fact that its support is included in the subdifferential of a convex function that we denote $\varphi(.,z)$. The measurability claim in Theorem 2.1, amounts to proving that one can select a $Q(u,z)\in \partial_u \varphi(u,z)$ in a jointly measurable way (note that since $\mu$ is absolutely continuous with respect to the Lebesgue measure, $\varphi(u,.)$ is single valued on a set of full measure for $\mu$ but which depends on $z$). Let us first define $q(u,z)$ as the conditional expectation of $\gamma(\nu^z)$ given $u$, then $q(u,z)$ is measurable separately in the variables $u$ and $z$. We know that $q(.,z)$ is a selection of the subgradient of a convex function $\varphi(.,z)$ which we may reconstruct by integration on the segment $[0,u]$ (there is no loss of generality in assuming here that $0$ lies in the interior of the support of $\mu$ and that $\varphi(0,z)=0$ for every $z$)  as $\varphi(u,z)=\int_0^1 q(tu,z) \cdot u dt$, from which it is easy to see that $\varphi$ is convex in $u$ and measurable in $z$. Now we regularize $\varphi$ by Yosida regularization i.e. set for every $n\in \NNN^*$
\[\varphi_n(u,z):=\inf_{v} \left \{\frac{n}{2} \vert u-v\vert^2 +\varphi(v,z) \right \}\]
for each $n$, $\varphi_n$ has a Lipschitz gradient with respect to $u$ and depends in a measurable way on $z$, hence $q_n(u,z):=\nabla_u \varphi_n(u,z)$ is Lipschitz in its first argument and measurable in its second one, it thus follows form Proposition 1.1 in chapter VIII of \cite{ektem} that $q_n$ is jointly measurable in $(u,z)$. It is well known that $q_n(u,z)$ converges pointwise to the element $Q(u,z)$ of minimal norm of $\partial_u \varphi(u,z)$ so that $Q$ is measurable.

Finally, in the case where one drops the finiteness of the second moments assumption, one cannot use optimal transport and $W_2$ as above. Instead, one has to proceed by approximation  as was done in the seminal work of McCann \cite{McCann1985}. First one has to suitably extend $\Gamma$  to the case of an arbitrary probability measure $\nu$ and again select a Borel (with respect to the narrow topology) selection of $\Gamma$, which can be done since the narrow topology is metrizable and separable on the set of Borel probability measures. The remainder of the argument is the same as above.

\section*{Acknowledgements}

We thank the editor Rungzhe Li, an associate editor, two anonymous referees, Denis Chetverikov, Xuming He, Roger Koenker, Steve Portnoy, and workshop participants at Oberwolfach 2013, Universit\'{e} Libre de Bruxelles, UCL, Columbia University, University of Iowa, MIT, and CORE for helpful comments. Jinxin He, J\'{e}r\^{o}me Santoul, and
Simon Weber provided excellent research assistance.

\end{document}